\RequirePackage{amsmath}

\documentclass{article}

\usepackage{amsthm}

\usepackage{xspace}
\usepackage[inline]{enumitem}
\usepackage{todonotes}

\usepackage{hyperref}

\newtheorem{theorem}{Theorem}
\newtheorem{lemma}{Lemma}

\binoppenalty=\maxdimen
\relpenalty=\maxdimen

\DeclareMathOperator{\skel}{skel}

\DeclareMathOperator{\corr}{corr}
\DeclareMathOperator{\apex}{apex}
\DeclareMathOperator{\spc}{space}
\DeclareMathOperator{\height}{height}

\newcommand{\ishape}{\textsf I~shape\xspace}
\newcommand{\eref}{e_{\mathrm{ref}}}
\newcommand{\muref}{\mu_{\mathrm{ref}}}

\title{An SPQR-Tree-Like Embedding Representation for Level Planarity}
\author{Guido Br\"uckner \and Ignaz Rutter}
\date{}

\begin{document}

\maketitle

\begin{abstract}
    An SPQR-tree is a data structure that efficiently represents all planar embeddings of a biconnected planar graph.
    It is a key tool in a number of constrained planarity testing algorithms, which seek a planar embedding of a graph subject to some given set of constraints.

    We develop an SPQR-tree-like data structure that represents all level-planar embeddings of a biconnected level graph with a single source, called the LP-tree, and give a simple algorithm to compute it in linear time.
    Moreover, we show that LP-trees can be used to adapt three constrained planarity algorithms to the level-planar case by using them as a drop-in replacement for SPQR-trees.
\end{abstract}

\section{Introduction}

Testing planarity of a graph and finding a planar embedding, if one exists, are classical algorithmic problems.
For visualization purposes, it is often desirable to draw a graph subject to certain additional constraints, e.g.,
finding orthogonal drawings~\cite{t-oeagitgwtmnob-87} or symmetric drawings~\cite{hme-altafcmssldotpg-06}, or inserting an edge into an embedding so that few edge crossings are caused~\cite{gmw-iaeiapg-05}.
Historically, these problems have been considered for embedded graphs.
More recent research has attempted to optimize not only one fixed embedding, but instead to optimize across all possible planar embeddings of a graph.
This includes
\begin{enumerate*}[label=(\roman*)]
    \item orthogonal drawings~\cite{brw-oogdwcdc-16},
    \item simultaneous embeddings, where one seeks to embed two planar graphs that share a common subgraph such that they induce the same embedding on the shared subgraph (see~\cite{bkr-sepg-13} for a survey),
    \item simultaneous orthogonal drawings~\cite{accddekklr-sop-16},
    \item embeddings where some edge intersections are allowed~\cite{ab-hpp-19},
    \item inserting an edge~\cite{gmw-iaeiapg-05}, a vertex~\cite{cgmw-iaviapg-09}, or multiple edges~\cite{ch-imeiapg-16} into an embedding,
    \item partial embeddings, where one insists that the embedding extends a given embedding of a subgraph~\cite{adfjk-tppeg-15}, and
    \item finding minimum-depth embeddings~\cite{adp-famdeoapgiot-11,bm-otcoepgtmcdm-90}.
\end{enumerate*}

The common tool in all of these recent algorithms is the SPQR-tree data structure, which efficiently represents all planar embeddings of a biconnected planar graph~$G$ by breaking down the complicated task of choosing a planar embedding of~$G$ into the task of independently choosing a planar embedding for each triconnected component of~$G$~\cite{dbt-ipt-89,dbt-olgawst-90,dt-omtcs-96,ht-dagitc-73,ml-ascopcg-37,t-cig-66}.
This is a much simpler task since the triconnected components have a very restricted structure, and so the components offer only basic, well-structured choices.

An \emph{upward planar drawing} is a planar drawing where each edge is represented by a~$y$-monotone curve.
For a level graph~$G=(V,E)$, which is a directed graph where each vertex~$v \in V$ is assigned to a level~$\ell(v)$ such that for each edge~$(u, v) \in E$ it is~$\ell(u) < \ell(v)$, a \emph{level-planar drawing} is an upward planar drawing where each vertex~$v$ is mapped to a point on the horizontal line~$y = \ell(v)$.
Level planarity can be tested in linear time~\cite{fpss-htmd-11,jl-lpelt-02,jlm-lptilt-98,rsbhk-sfplg-01}.
Recently, the problem of extending partial embeddings for level-planar drawings has been studied~\cite{br-pclp-17}.
While the problem is \textsf{NP}-hard in general, it can be solved in polynomial time for single-source graphs.
Very recently, an SPQR-tree-like embedding representation for upward planarity has been used to extend partial upward embeddings~\cite{bhr-astlerfup-19}.
The construction crucially relies on an existing decomposition result for upward planar graphs~\cite{hl-updossad-96}.
No such result exists for level-planar graphs.
Moreover, the level assignment leads to components of different ``heights'', which makes our decompositions significantly more involved.

\paragraph{Contribution.}
We develop the LP-tree, an analogue of SPQR-trees for level-planar embeddings of level graphs with a single source whose underlying undirected graph is biconnected.
It represents the choice of a level-planar embedding of a level-planar graph by individual embedding choices for certain components of the graph, for each of which the embedding is either unique up to reflection, or allows to arbitrarily permute certain subgraphs around two pole vertices.
Its size is linear in the size of~$G$ and it can be computed in linear time.
The LP-tree is a useful tool that unlocks the large amount of SPQR-tree-based algorithmic knowledge for easy translation to the level-planar setting.
In particular, we obtain linear-time algorithms for partial and constrained level planarity for biconnected single-source level graphs, which improves upon the~$O(n^2)$-time algorithm known to date~\cite{br-pclp-17}.
Further, we describe the first efficient algorithm for the simultaneous level planarity problem when the shared graph is a biconnected single-source level graph.


\section{Preliminaries}
\label{sec:preliminaries}

\addtocontents{toc}{\protect\setcounter{tocdepth}{1}}

Let~$G = (V, E)$ be a connected level graph.
For each vertex~$v \in V$ let~$d(v) \ge \ell(v)$ denote the \emph{demand} of~$v$.
An \emph{apex} of some vertex set~$V' \subseteq V$ is a vertex~$v \in V'$ whose level is maximum.
The \emph{demand} of~$V'$, denoted by~$d(V')$, is the maximum demand of a vertex in~$V'$.
An apex of a face~$f$ is an apex of the vertices incident to~$f$.
A \emph{planar drawing} of~$G$ is a topological planar drawing of the underlying undirected graph of~$G$.
Planar drawings are \emph{equivalent} if they can be continuously transformed into each other without creating intermediate intersections.
A \emph{planar embedding} is an equivalence class of equivalent planar drawings.
\paragraph{Level Graphs and Level-Planar Embeddings.}
A \emph{path} is a sequence of vertices~$(v_1, v_2, \ldots, v_j)$ so that for~$1 \leq i < j$ either~$(v_i, v_{i + 1})$ or~$(v_{i + 1}, v_i)$ is an edge in~$E$.
A \emph{directed path} is a sequence~$(v_1, v_2, \ldots, v_j)$ of vertices so that for~$1 \leq i < j$ it is~$(v_i, v_{i + 1}) \in E$.
A vertex~$u$ \emph{dominates} a vertex~$v$ if there exists a directed path from~$u$ to~$v$.
A vertex is a \emph{sink} if it dominates no vertex except for itself.
A vertex is a \emph{source} if it is dominated by no vertex except for itself.
An \emph{$st$-graph} is a graph with a single source and a single sink, usually denoted by~$s$ and~$t$, respectively.
Throughout this paper all graphs are assumed to have a single source~$s$.
%
For the remainder of this paper we restrict our considerations to level-planar drawings of~$G$ where each vertex~$v \in V$ that is not incident to the outer face is incident to some inner face~$f$ so that each apex~$a$ of the set of vertices on the boundary of~$f$ satisfies~$d(v) < \ell(a)$.
We will use demands in Section~\ref{sec:applications} to restrict the admissible embeddings of biconnected components in the presence of cutvertices.
Note that setting~$d(v) = \ell(v)$ for each~$v \in V$ gives the conventional definition of level-planar drawings.
A planar embedding~$\Gamma$ of~$G$ is \emph{level planar} if there exists a level-planar drawing of~$G$ with planar embedding~$\Gamma$.
We then call~$\Gamma$ a \emph{level-planar embedding}.
For single-source level graphs, level-planar embeddings are equivalence classes of topologically equivalent level-planar drawings.

\begin{lemma}
    The level-planar drawings of a single-source level graph correspond bijectively to its level-planar combinatorial embeddings.
    \label{lem:lp-embedding-is-equivalence-class}
\end{lemma}

\begin{proof}
    Let~$G = (V, E)$ be a single-source~$k$-level graph.
    Assume without loss of generality that~$G = (V, E)$ is proper, i.e., for each edge~$(u, v) \in E$ it is~$\ell(u) + 1 = \ell(v)$.
    Let~$u, v \in V_i$ be two vertices on level~$i$ with~$1 \leq i \leq k$.
    Further, let~$w$ be a vertex of~$G$ so that there are disjoint directed paths~$p_u$ and~$p_v$ from~$w$ to~$u$ and~$v$, respectively.
    Because~$G$ is a single-source graph, such a vertex must exist.
    Let~$e$ and~$f$ denote the first edge on~$p_u$ and~$p_v$, respectively.
    Further, let~$\prec$ be a level-planar drawing of~$G$ and let~$\mathcal G$ be a level-planar combinatorial embedding of~$G$.
    If~$w$ is not the single source of~$G$, it has an incoming edge~$g$.
    Then it is~$u \prec_i v$ if and only if~$e, f$ and~$g$ appear in that order around~$w$.
    Otherwise, if~$w$ is the source of~$G$, let~$g$ denote the edge~$(w, t)$, which exists by construction.
    Because~$g$ is embedded as the leftmost edge, it is~$u \prec_i v$ if and only if~$g, e$ and~$f$ appear in that order around~$w$.
    The claim then follows easily.
\end{proof}
To make some of the subsequent arguments easier to follow, we preprocess our input level graph~$G$ on~$k$ levels to a level graph~$G'$ on~$d(V) + 1$ levels as follows.
We obtain~$G'$ from~$G$ by adding a new vertex~$t$ on level~$d(V) + 1$ with demand~$d(t) = d(V) + 1$, connecting it to all vertices on level~$k$ and adding the edge~$(s, t)$.
Note that~$G'$ is generally not an~$st$-graph.
Let~$H$ be a graph with a level-planar embedding~$\Lambda$ and let~$H'$ be a supergraph of~$H$ with a level-planar embedding~$\Lambda'$.
The embedding~$\Lambda'$ \emph{extends}~$\Lambda$ when~$\Lambda'$ and~$\Lambda$ coincide on~$H$.
The embeddings of~$G'$ where the edge~$(s, t)$ is incident to the outer face and the embeddings of~$G$ are, in a sense, equivalent.

\begin{lemma}
    An embedding~$\Gamma$ of~$G$ is level-planar if and only if there exists a level-planar embedding~$\Gamma'$ of~$G'$ that extends~$\Gamma$ where~$(s, t)$ is incident to the outer face.
    \label{lem:t-extension}
\end{lemma}

\begin{proof}
    Let~$G = (V, E)$ be a~$k$-level graph, and let~$G'$ be the supergraph of~$G$ as described above together with a level-planar embedding~$\Gamma'$.
    Because~$G$ is a subgraph of~$G'$, restricting~$\Gamma'$ to~$G$ immediately gives a level-planar embedding~$\Gamma$ of~$G$ that is extended by~$\Gamma'$.

    Now let~$\Gamma$ be a level-planar embedding of~$G$.
    Since all apices of~$V$ lie on the outer face, the newly added vertex~$t$ can be connected to those vertices without causing any edge crossings.
    Then, because~$s$ is the single source of~$G$ and~$t$ is the sole apex of~$V(G')$, the edge~$(s, t)$ can be drawn into the outer face as a~$y$-monotone curve without causing edge crossings.
    Let~$\Gamma'$ refer to the resulting embedding.
    Then~$\Gamma'$ is a level-planar embedding of~$G'$ that extends~$\Gamma$.
\end{proof}

To represent all level-planar embeddings of~$G$, it is sufficient to represent all level-planar embeddings of~$G'$ and remove~$t$ and its incident edges from all embeddings.
It is easily observed that if~$G$ is a biconnected single-source graph, then so is~$G'$.
We assume from now on that the vertex set of our input graph~$G$ has a unique apex~$t$ and that~$G$ contains the edge~$(s, t)$.
We still refer to the highest level as level~$k$, i.e., the apex~$t$ lies on level~$k$.

Level-planar embeddings of a graph have an important relationship with level-planar embeddings of~$st$-supergraphs thereof.
We use Lemmas~\ref{lem:st-augmentation} and~\ref{lem:st-lp-is-p}, and a novel characterization of single-source level planarity in Lemma~\ref{lem:level-planarity-characterisation} to prove that certain planar embeddings are also level planar.

\begin{lemma}
    Let~$G = (V, E)$ be a single-source level graph with a unique apex.
    Further, let~$\Gamma$ be a level-planar embedding of $G$.
    Then there exists an~$st$-graph~$G_{st} = (V, E \cup E_{st})$ together with a level-planar embedding~$\Gamma_{st}$ that extends~$\Gamma$.
    \label{lem:st-augmentation}
\end{lemma}
\begin{proof}
    We prove the claim by induction over the number of sinks in~$G$.
    Note that because~$t$ is an apex of~$G$, it must be a sink.
    So~$G$ has at least one sink.
    If~$G$ has one sink, the claim is trivially true for~$E_{st} = \emptyset$.
    Now suppose that~$G$ has more than one sink.
    Let~$w \neq t$ be a sink of~$G$.
    In some drawing of~$G$ with embedding~$\Gamma$, walk up vertically from~$w$ into the incident face above~$w$.
    If a vertex~$v$ or an edge~$(u, v)$ is encountered, set~$E_{st} = \{(w, v)\}$.
    If no vertex or edge is encountered,~$w$ lies on the outer face of~$\Gamma$.
    Then set~$E_{st} = \{(w, t)\}$.
    Note that in both cases the added edges can be embedded into~$\Gamma$ as~$y$-monotone curves while maintaining level planarity.
    Then extend~$E_{st}$ inductively, which shows the claim.
\end{proof}

Next we establish a characterization of the planar embeddings that are level planar.
The following lemma is implicit in the planarity test for~$st$-graphs by Chiba~\cite{cnao-alafepgup-85} and the work on upward planarity by Di Battista and Tamassia~\cite{dt-afproad-88}.
\begin{lemma}
    Let~$G$ be an~$st$-graph.
    Then each planar embedding~$\Gamma$ of~$G$ is also a level-planar embedding of~$G$ in which~$(s, t)$ is incident to the outer face, and vice versa.
    \label{lem:st-lp-is-p}
\end{lemma}
\begin{proof}
    Consider a vertex~$v \neq s, t$ of~$G$.
    Then the incoming and outgoing edges appear consecutively around~$v$ in~$\Gamma$.
    To see this, suppose that there are four vertices~$w, x, y, z \in V$ with edges~$(w, v), (v, x), (y, v), (v, z) \in E$ that appear in that counter-clockwise cyclic order around~$v$ in~$\Gamma$.
    Because~$G$ is an~$st$-graph there are directed paths~$p_w$ and~$p_y$ from~$s$ to~$w$ and~$y$, respectively, and directed paths~$p_x$ and~$p_z$ from~$x$ and~$z$ to~$t$, respectively.
    Moreover,~$p \in \{ p_w, p_y \}$ and~$p' \in \{ p_x, p_z \}$ are disjoint and do not contain~$v$.
    Then some~$p \in \{ p_w, p_y \}$ and~$p' \in \{ p_x, p_z \}$ must intersect, a contradiction to the fact that~$\Gamma$ is planar; see~Fig.~\ref{fig:st-lp-is-p}~(a).

    Let~$e_1, e_2, \dots, e_i, e_{i + 1}, \dots, e_n$ denote the counter-clockwise cyclic order of edges around~$v$ in~$\Gamma$ so that~$e_1, \dots, e_i$ are incoming edges and~$e_{i + 1}, \dots, e_n$ are outgoing edges.
    Let~$e_1, \dots, e_i$ denote the left-to-right order of incoming edges and let~$e_n, e_{n - 1}, \dots, e_{i + 1}$ denote the left-to-right order of outgoing edges.
    Split the clockwise cyclic order of edges around~$s$ at~$(s, t)$ to obtain the left-to-right order of outgoing edges.
    Symmetrically, split counter-clockwise order of edges around~$t$ at~$(s, t)$ to obtain the left-to-right order of incoming edges.

    Create a level-planar embedding~$\Gamma'$ of~$G$ step by step as follows; see Fig.~\ref{fig:st-lp-is-p}.
    Draw vertices~$s$ and~$t$ on levels~$\ell(s)$ and~$\ell(t)$, respectively, and connect them by a straight line segment.
    Call the vertices~$s, t$ and the edge~$(s, t)$ \emph{discovered}.
    Call the path~$s, t$ the \emph{right frontier}.
    Call a vertex on the right frontier \emph{settled} if all of its outgoing edges are discovered.

    More generally, let~$s = u_1, u_2, \dots, u_n = t$ denote the right frontier.
    Modify the right frontier while maintaining that
    \begin{enumerate*}[label=(\roman*)]
        \item the right frontier is a directed path from~$s$ to~$t$,
        \item any edge~$(u_a, u_{a + 1})$ on the right frontier is the rightmost discovered outgoing edge around~$u_a$, and
        \item the right frontier is incident to the outer face of~$\Gamma'$.
    \end{enumerate*}

    Let~$u_i$ denote the vertex on the right frontier closest to~$t$ that is not settled.
    Discover the leftmost undiscovered outgoing edges starting from~$u_i$ to construct a directed path~$v_1 = u_i, v_2, \dots, v_m$, where~$v_m$ is the first vertex that had been discovered before.
    Because~$G$ has a single sink such a vertex exists.
    Because~$\Gamma$ is planar~$v_m$ lies on the right frontier, i.e.,~$v_m = u_j$ for some~$j$ with~$i < j \le n$.
    Insert the vertices~$v_2, \dots, v_{m - 1}$ and the edges~$(v_a, v_{a + 1})$ for~$1 \le a < m$ to the right of the path~$u_i, \dots, u_j$ into~$\Gamma'$ (Property~(iii) of the invariant), maintaining level planarity of~$\Gamma'$.
    This creates a new face~$f$ of~$\Gamma'$ whose boundary is~$u_i, u_{i + 1}, \dots, u_j = v_m, v_{m - 1}, \dots, v_1 = u_i$.

    We show that~$f$ is a face of~$\Gamma$.
    Because~$u_a$ is settled there cannot be an undiscovered outgoing edge between~$(u_{a - 1}, u_a)$ and~$(u_a, u_{a + 1})$ in the counter-clockwise order of edges around~$u_a$ in~$\Gamma$ for~$i < a < j$ (see edge~$g$ in Fig.~\ref{fig:st-lp-is-p}~(b)).
    There can also not be a discovered outgoing edge because of Property~(ii) of the invariant (see edge~$e$ in Fig.~\ref{fig:st-lp-is-p}~(b)).
    Because the leftmost undiscovered edge is chosen there is no undiscovered outgoing edge between~$(v_a, v_{a + 1})$ and~$(v_{a - 1}, v_a)$ in the counter-clockwise order of edges around~$v_a$ in~$\Gamma$ for~$1 < a < m$ (see edge~$h$ in Fig.~\ref{fig:st-lp-is-p}~(b)).
    There can also not be a discovered outgoing edge because~$v_a$ was not discovered before (see edge~$q$ in Fig.~\ref{fig:st-lp-is-p}~(b)).
    There can be no outgoing edge between~$(v_1, v_2)$ and~$(u_i, u_{i + 1})$ in the counter-clockwise order of edges around~$v_1 = u_i$ because either such an edge would be discovered contradicting Property~(ii), or not, contradicting the fact that~$(v_1, v_2)$ is chosen as the leftmost undiscovered outgoing edge of~$v_1$.
    There can be no outgoing edge between~$(u_{j - 1}, u_j)$ and~$(v_{m - 1}, v_m)$ in the counter-clockwise order of edges around~$u_j = v_m$ because either~$u_j = v_m = t$ is a sink, or the incoming and outgoing edges appear consecutively around~$u_j = v_m$ in~$\Gamma$ (see edge~$d$ in Fig.~\ref{fig:st-lp-is-p}~(b)).

    There can also be no incoming edge~$(u, v)$ between any of these edge pairs (see edge~$r$ in Fig.~\ref{fig:st-lp-is-p}~(b)).
    This is because~$G$ has a single source~$s$, so there exists a directed path~$p$ from~$s$ to~$u$.
    Because~$u$ lies inside of~$f$ the path~$p$ must contain a vertex~$x$ on the boundary of~$f$.
    Then~$p$ would also contain an outgoing edge of~$x$ which we have just shown to be impossible.

    Let~$s = u_1, u_2, \dots, u_i = v_1, v_2, \dots, v_m = u_j, \dots, u_n = t$ denote the new right frontier.
    Note that the invariant holds for this modified right frontier.
    Because~$G$ has a single-source all vertices and edges are drawn in this way.
    Because~$\Gamma$ and~$\Gamma'$ have the same faces they are the same embedding.
    Finally,~$\Gamma'$ is level planar by construction, which shows the claim.
    \begin{figure}[t]
        \centering
        \includegraphics{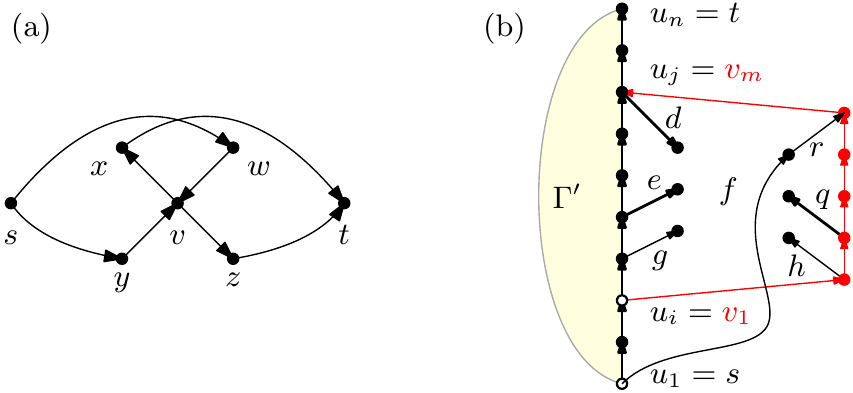}
        \caption{
            Proof of Lemma~\ref{lem:st-lp-is-p}.
            The incoming and outgoing edges around each vertex are consecutive~(a).
            Creating the level-planar embedding~$\Gamma'$ by attaching the path~$v_1, v_2, \dots, v_m$ (drawn in red) to the right frontier~$u_1, u_2, \dots, u_n$, thereby creating a new face~$f$.
            Discovered edges are drawn thickly.
            The edges~$e, g, h, q, r, d$ cannot exist.
        }
        \label{fig:st-lp-is-p}
    \end{figure}
\end{proof}

Thus, a planar embedding~$\Gamma$ of a graph~$G$ is level-planar if and only if it can be augmented to an~$st$-graph~$G' \supseteq G$ such that all augmentation edges can be embedded in the faces of~$\Gamma$ without crossings.
This gives rise to the following characterization.

\begin{lemma}
    \label{lem:level-planarity-characterisation}
    Let~$G$ be a single-source~$k$-level graph with a unique apex~$t$.
    Then~$G$ is level planar if and only if it has a planar embedding where every vertex~$v$ with~$\ell(v) < k$ is incident to at least one face~$f$ so that~$v$ is not an apex of~$f$.
\end{lemma}
\begin{proof}
    Let~$\Gamma_l$ be a level-planar drawing of~$G$.
    Consider a vertex~$v$ such that it is~$\ell(v) < \ell(t)$.
    If~$v$ has an outgoing edge~$(v, w)$, then~$v$ and~$w$ are incident to some shared face~$f$.
    Because it is~$\ell(v) < \ell(w)$, vertex~$v$ is not an apex of~$f$.
    If~$v$ has no outgoing edges, start walking upwards from~$v$ in a straight line.
    Stop walking upwards if an edge~$(u, w)$ or a vertex~$w$ is encountered.
    Then~$v$ and~$w$ are again incident to some shared face~$f$.
    Moreover, it is~$\ell(v) < \ell(w)$, and therefore~$v$ is not an apex of~$f$.
    If no edge or vertex is encountered when walking upwards,~$v$ must lie on the outer face.
    Because~$t$ lies on the outer face and it is~$\ell(v) < \ell(t)$, vertex~$v$ is not an apex of the outer face.
    Finally, because~$\Gamma_l$ is level planar it is, of course, also planar.

    Now let~$\Gamma_p$ be a planar embedding of~$G$.
    The idea is to augment~$G$ and~$\Gamma_p$ by inserting edges so that~$G$ becomes an~$st$-graph together with a planar embedding~$\Gamma_p$.
    To that end, consider a sink~$v \neq t$ of~$G$.
    By assumption,~$v$ is incident to at least one face~$f$ so that~$v$ is not an apex of~$f$.
    Hence, it is~$\ell(v) < \ell(\apex(f))$.
    So the augmentation edge~$e = (v, \apex(f))$ can be inserted into~$G$ without creating a cycle.
    Further,~$e$ can be embedded into~$f$.
    Because all augmentation edges embedded into~$f$ have endpoint~$\apex(f)$, the embedding~$\Gamma_p$ of~$G$ remains planar.
    This means that~$G$ can be augmented so that~$t$ becomes the only sink while maintaining the planarity of~$\Gamma_p$.
    Because~$G$ also has a single source,~$G$ is now an~$st$-graph and it follows from Lemma~\ref{lem:st-lp-is-p} that~$\Gamma_p$ is not only planar, but also level planar.
\end{proof}

\paragraph{Decomposition Trees and SPQR-Trees.}

Our description of decomposition trees follows Angelini et al.~\cite{abr-tmdopg-14}.
Let~$G$ be a biconnected graph.
A \emph{separation pair} is a subset~$\{u, v\} \subseteq V$ whose removal from~$G$ disconnects~$G$.
%
Let~$\{u, v\}$ be a separation pair and let~$H_1, H_2$ be two subgraphs of~$G$ with~$H_1 \cup H_2 = G$ and~$H_1 \cap H_2 = \{u, v\}$.
Define the tree~$\mathcal T$ that consists of two \emph{nodes}~$\mu_1$ and~$\mu_2$ connected by an undirected \emph{arc} as follows.
For~$i = 1, 2$ node~$\mu_i$ is equipped with a multigraph~$\skel(\mu_i) = H_i + e_i$, called its \emph{skeleton}, where~$e_i = (u, v)$ is called a \emph{virtual} edge.
The arc~$(\mu_1,\mu_2)$ links the two virtual edges~$e_i$ in~$\skel(\mu_i)$ with each other.
We also say that the virtual edge~$e_1$ \emph{corresponds} to~$\mu_2$ and likewise that~$e_2$ corresponds to~$\mu_1$.
The idea is that~$\skel(\mu_1)$ provides a more abstract view of~$G$ where~$e_1$ serves as a placeholder for~$H_2$.
More generally, there is a bijection~$\corr_\mu \colon E(\skel(\mu)) \to N(\mu)$ that maps every virtual edge of~$\skel(\mu)$ to a neighbor of~$\mu$ in~$\mathcal T$, and vice versa.
If it is~$\corr_\mu((u, v)) = \nu$, then~$\nu$ is said to have \emph{poles}~$u$ and~$v$ in~$\mu$.
If~$\mu$ is clear from the context we simply say that~$\nu$ has poles~$u, v$.
When the underlying graph is a level graph, we assume~$\ell(u) \le \ell(v)$ without loss of generality.
For an arc~$(\nu,\mu)$ of~$\mathcal T$, the virtual edges~$e_1,e_2$ with~$\corr_\mu(e_1) = \nu$ and~$\corr_\nu(e_2) = \mu$ are called \emph{twins}, and~$e_1$ is called the \emph{twin} of~$e_2$ and vice versa.
This procedure is called a \emph{decomposition}, see Fig.~\ref{fig:decomposition-tree} on the left.
It can be re-applied to skeletons of the nodes of~$\mathcal T$, which leads to larger trees with smaller skeletons.
A tree obtained in this way is a \emph{decomposition tree} of~$G$.
\begin{figure}[t]
    \centering
    \includegraphics[width=\linewidth]{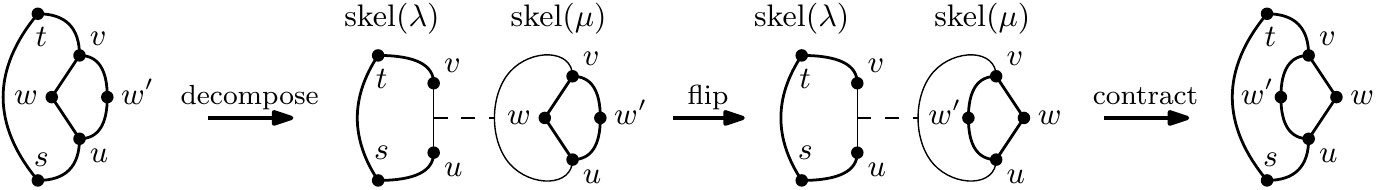}
    \caption{
        Decompose the embedded graph~$G$ on the left at the separation pair~$u, v$.
        This gives the center-left decomposition tree whose skeletons are embedded as well.
        Reflecting the embedding of~$\skel(\mu)$ or, equivalently, flipping~$(\lambda, \mu)$, yields the same decomposition tree with a different embedding of~$\skel(\mu)$.
        Contract~$(\lambda, \mu)$ to obtain the embedding on the right.
    }
    \label{fig:decomposition-tree}
\end{figure}
A decomposition can be undone by \emph{contracting} an arc~$(\mu_1,\mu_2)$ of~$\mathcal T$, forming a new node~$\mu$ with a larger skeleton as follows.
Let~$e_1, e_2$ be twin edges in~$\skel(\mu_1), \skel(\mu_2)$.
The skeleton of~$\mu$ is the union of~$\skel(\mu_1)$ and~$\skel(\mu_2)$ without the two twin edges~$e_1,e_2$.
Contracting all arcs of a decomposition tree of~$G$ results in a decomposition tree consisting of a single node whose skeleton is~$G$.
See Fig.~\ref{fig:decomposition-tree} on the right.
Let~$\mu$ be a node of a decomposition tree with a virtual edge~$e$ with~$\corr_\mu(e) = \nu$.
The \emph{expansion graph} of~$e$ and~$\nu$ in~$\mu$, denoted by~$G(e)$ and~$G(\mu, \nu)$, respectively, is the graph obtained by removing the twin of~$e$ from~$\skel(\nu)$ and contracting all arcs in the subtree that contains~$\nu$.

Each skeleton of a decomposition tree of~$G$ is a minor of~$G$.
So if~$G$ is planar each skeleton of a decomposition tree~$\mathcal T$ of~$G$ is planar as well.
If~$(\mu_1,\mu_2)$ is an arc of~$\mathcal T$, and~$\skel(\mu_1)$ and~$\skel(\mu_2)$ have fixed planar embeddings~$\Gamma_1$ and~$\Gamma_2$, respectively, then the skeleton of the node~$\mu$ obtained from contracting~$(\mu_1,\mu_2)$ can be equipped with an embedding~$\Gamma$ by merging these embeddings along the twin edges corresponding to~$(\mu_1,\mu_2)$; see Fig.~\ref{fig:decomposition-tree} center.
This requires at least one of the virtual edges~$e_1$ in~$\skel(\mu_1)$ with~$\corr_{\mu_1}(e_1) = \mu_2$ or~$e_2$ in~$\skel(\mu_2)$ with~$\corr_{\mu_2}(e_2) = \mu_1$ to be incident to the outer face.
If we equip every skeleton with a planar embedding and contract all arcs, we obtain a planar embedding of~$G$.
This embedding is independent of the order of the edge contractions.
Thus, every decomposition tree~$\mathcal T$ of~$G$ represents (not necessarily all) planar embeddings of~$G$ by choosing a planar embedding of each skeleton and contracting all arcs.
Let~$\eref$ be an edge of~$G$.
Rooting~$\mathcal T$ at the unique node~$\muref$ whose skeleton contains the real edge~$\eref$ identifies a unique parent virtual edge in each of the remaining nodes; all other virtual edges are called \emph{child virtual edges}.
The arcs of~$\mathcal T$ become directed from the parent node to the child node.
Restricting the embeddings of the skeletons so that the parent virtual edge (the edge~$\eref$ in case of~$\muref$) is incident to the outer face, we obtain a representation of (not necessarily all) planar embeddings of~$G$ where~$\eref$ is incident to the outer face.
Let~$\mu$ be a node of~$\mathcal T$ and let~$e$ be a child virtual edge in~$\skel(\mu)$ with~$\corr_\mu(e) = \nu$.
Then the expansion graph~$G(\mu, \nu)$ is simply referred to as~$G(\nu)$.

The \emph{SPQR-tree} is a special decomposition tree whose skeletons are precisely the triconnected components of~$G$.
It has four types of nodes: S-nodes, whose skeletons are cycles, P-nodes, whose skeletons consist of three or more parallel edges between two vertices, and R-nodes, whose skeletons are simple triconnected graphs.
Finally, a Q-node has a skeleton consisting of two vertices connected by one real and by one virtual edge.
This means that in the skeletons of all other node types all edges are virtual.
In an SPQR-tree the embedding choices are of a particularly simple form.
The skeletons of Q- and S-nodes have a unique planar embedding (not taking into account the choice of the outer face).
The child virtual edges of P-node skeletons may be permuted arbitrarily, and the skeletons of R-nodes are 3-connected, and thus have a unique planar embedding up to reflection.
We call this the \emph{skeleton-based} embedding representation.
There is also an \emph{arc-based} embedding representation.
Here the embedding choices are (i) the linear order of the children in each P-node, and (ii) for each arc~$(\lambda, \mu)$ whose target~$\mu$ is an R-node whether the embedding of the expansion graph~$G(\mu)$ should be \emph{flipped}.
To obtain the embedding of~$G$, we contract the edges of~$\mathcal T$ bottom-up.
Consider the contraction of an arc~$(\lambda,\mu)$ whose child~$\mu$ used to be an R-node in~$\mathcal T$.
At this point,~$\skel(\mu)$ is equipped with a planar embedding~$\Gamma_\mu$.
If the embedding should be flipped, we reflect the embedding~$\Gamma_\mu$ before contracting~$(\lambda,\mu)$, otherwise we simply contract~$(\lambda,\mu)$.
The arc-based and the skeleton-based embedding representations are equivalent.
See Fig.~\ref{fig:spqr-tree} and Fig.~\ref{fig:lp-tree}~(a,b) for examples of a planar graph and its SPQR-tree.
\begin{figure}
    \centering
    \includegraphics[width=\linewidth]{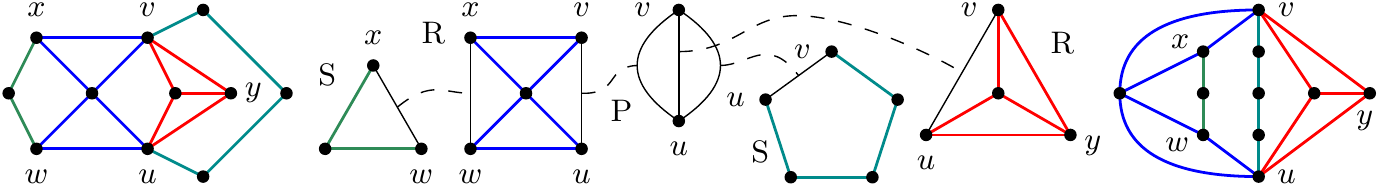}
    \caption{
        A planar graph on the left and its SPQR-tree in the middle.
        The five nodes of the SPQR-tree are represented by their respective skeleton graphs.
        Dashed edges connect twin virtual edges and colored edges correspond to Q-nodes.
        The embedding of the graph on the right is obtained by flipping the embedding of the blue R-node and swapping the middle and right edge of the P-node.
    }
    \label{fig:spqr-tree}
\end{figure}

\section{A Decomposition Tree for Level Planarity}
\label{sec:decomposition}

We construct a decomposition tree of a given single-source level graph~$G$ whose underlying undirected graph is biconnected that represents all level-planar embeddings of~$G$, called the \emph{LP-tree}.
As noted in the Preliminaries, we assume that~$G$ has a unique apex~$t$, for which~$\ell(t) = d(t)$ holds true.
The LP-tree for~$G$ is constructed based on the SPQR-tree for~$G$.
We keep the notion of S-, P-, Q- and R-nodes and construct the LP-tree so that the nodes behave similarly to their namesakes in the SPQR-tree.
The skeleton of a P-node consists of two vertices that are connected by at least three parallel virtual edges that can be arbitrarily permuted.
The skeleton of an R-node~$\mu$ is equipped with a \emph{reference embedding}~$\Gamma_\mu$, and the choice of embeddings for such a node is limited to either~$\Gamma_\mu$ or its reflection.
Unlike in SPQR-trees, the skeleton of~$\mu$ need not be triconnected, instead it can be an arbitrary biconnected planar graph.
The embedding of R-node skeletons being fixed up to reflection allows us to again use the equivalence of the arc-based and the skeleton-based embedding representations.

The construction of the LP-tree starts out with an SPQR-tree~$\mathcal T$ of~$G$.
Explicitly label each node of~$\mathcal T$ as an S-, P-, Q- or R-node.
This way, we can continue to talk about S-, P-, Q- and R-nodes of our decomposition tree even when they no longer have their defining properties in the sense of SPQR-trees.
Assume that the edge~$(s, t)$ to be incident to the outer face of every level-planar drawing of~$G$ (Lemma~\ref{lem:t-extension}), i.e., consider~$\mathcal T$ rooted at the Q-node corresponding to~$(s, t)$.
The construction of our decomposition tree works in two steps.
First, decompose the graph further by decomposing P-nodes in order to disallow permutations that lead to embeddings that are not level planar.
Second, contract arcs of the decomposition tree, each time fixing a reference embedding for the resulting node, so that we can consider it as an R-node, such that the resulting decomposition tree represents exactly the level-planar embeddings of~$G$.
The remainder of this section is structured as follows.
The details and correctness of the first step are given in Section~\ref{ssec:p-node-splits}.
Section~\ref{ssec:arc-processing} gives the algorithm for constructing the final decomposition tree~$\mathcal T$.
It follows from the construction that all embeddings it represents are level-planar, and Section~\ref{sssec:correctness-arc-processing} shows that, conversely, it also represents every level-planar embedding.
In Section~\ref{ssec:implementation-in-linear-time}, we present a linear-time implementation of the construction algorithm.

\subsection{P-Node Splits}
\label{ssec:p-node-splits}

In SPQR-trees, the children of P-nodes can be arbitrarily permuted.
We would like P-nodes of the LP-tree to have the same property.
Hence, we decompose skeletons of P-nodes to disallow orders that lead to embeddings that are not level planar.
The decomposition is based on the height of the child virtual edges, which we define as follows.
Let~$\mu$ be a node of a rooted decomposition tree and let~$u$ and~$v$ be the poles of~$\mu$.
Define~$V(\mu) = V(G(\mu)) \setminus \{u, v\}$.
The \emph{height} of~$\mu$ and of the child virtual edge~$e$ with~$\corr(e) = \mu$ is~$d(\mu) = d(e) = d(V(\mu))$.

Now let~$\mu$ be a P-node, and let~$\Gamma$ be a level-planar embedding of~$G$.
The embedding~$\Gamma$ induces a linear order of the child virtual edges of~$\mu$.
This order can be obtained by splitting the combinatorial embedding of~$\skel(\mu)$ around~$u$ at the parent edge.
Then the following is true.

\begin{lemma}
    \label{lem:p-node-heights}
    Let~$\mathcal T$ be a decomposition tree of~$G$, let~$\mu$ be a P-node of~$\mathcal T$ with poles~$u, v$, and let~$e_{\max}$ be a child virtual edge of~$\mu$ with maximal height.
    Further, let~$\Gamma$ be a level-planar embedding of~$G$ that is represented by~$\mathcal T$.
    If the height of~$e_{\max}$ is at least~$\ell(v)$, then~$e_{\max}$ is either the first or the last edge in the linear ordering of the child virtual edges induced by~$\Gamma$.
\end{lemma}

\begin{proof}
    Let~$\nu = \corr_\mu(e_{\max})$.
    Further, let~$G_{\max} = G(e_{\max})$, and let~$w \in V(\nu)$ with~$d(w) = d(\nu)$.
    If~$d(w) < \ell(v)$, the statement of the lemma is trivially satisfied, so assume~$d(w) \ge \ell(v)$ and suppose that~$e_{\max}$ is not the first edge or last edge.
    Let~$\Gamma_\mu$ be the embedding of~$\skel(\mu)$ in the corresponding skeleton-based representation of~$\Gamma$.
    Then there are child virtual edges~$e_1,e_2$ immediately preceding and succeeding edge~$e_{\max}$ in~$\Gamma_\mu$, respectively.
    By construction of the embedding~$\Gamma$ via contractions from the embeddings of skeletons, it follows that~$w$ shares a face only with the inner vertices of~$G(e_i)$ for~$i=1,2$, the inner vertices of~$G_{\max}$, and~$u$ and~$v$.
    By the choice of~$e_{\max}$ it follows that~$d(w) \ge \ell(w')$ for all inner vertices~$w'$ of~$G(e_i)$,~$i=1,2$, and the choice of~$w$ guarantees that~$d(w) \ge \ell(w')$ for all inner vertices~$w'$ of~$G(e_{\max})$.
    Moreover, it is~$d(w) \ge \ell(v) \ge \ell(u)$ by assumption.
    It follows that~$w$ is not incident to any face that has an apex~$a$ with~$d(w) < \ell(a)$.
    Beause~$w$ is an inner vertex of~$G_{\max}$ it is not incident to the outer face.
    Thus,~$\Gamma$ is not level-planar by Lemma~\ref{lem:level-planarity-characterisation}, a contradiction.
\end{proof}

Lemma~\ref{lem:p-node-heights} motivates the following modification of a decomposition tree~$\mathcal T$.
Take a P-node~$\mu$ with poles~$u, v$ that has a child edge whose height is at least~$\ell(v)$.
Denote by~$\lambda$ the parent of~$\mu$.
Further, let~$e_{\max}$ be a child virtual edge with maximum height and let~$e_{\mathrm{parent}}$ denote the parent edge of~$\skel(\mu)$.
Obtain a new decomposition tree~$\mathcal T'$ by splitting~$\mu$ into two nodes~$\mu_1$ and~$\mu_2$ representing the subgraph~$H_1$ consisting of the edges~$e_{\max}$ and~$e_{\mathrm{parent}}$, and the subgraph~$H_2$ consisting of the remaining child virtual edges, respectively; see Fig.~\ref{fig:chain}.
\begin{figure}[t]
    \centering
    \includegraphics{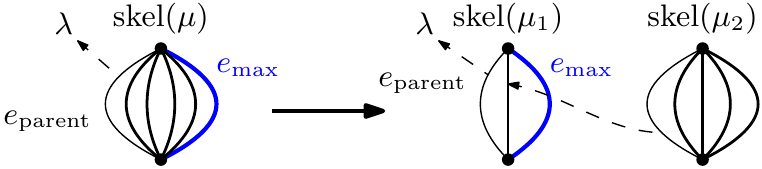}
    \caption{
        Result of a P-node~$\mu$ split with parent~$\lambda$ and child with maximum height~$\nu$.
        Note that after the split,~$\mu_1$ is an R-node and~$\mu_2$ has one less child than~$\mu$ had.
    }
    \label{fig:chain}
\end{figure}
Note that the skeleton of~$\mu_1$, which corresponds to~$H_1$, has only two child virtual edges.
We therefore define it to be an R-node.
Moreover, observe that in any embedding of~$\skel(\mu)$ that is obtained from choosing embeddings for~$\skel(\mu_1)$ and~$\skel(\mu_2)$ and contracting the arc~$(\mu_1,\mu_2)$, the edge~$e_{\max}$ is the first or last child edge.
Conversely, because~$\mu_2$ is a P-node, all embeddings where~$e_{\max}$ is the first or last child edge are still represented by~$\mathcal T'$.
Apply this decomposition iteratively, creating new R-nodes on the way, until each P-node~$\mu$ with poles~$u$ and~$v$ has only child virtual edges~$e$ that have height at most~$\ell(v) - 1$.
We say that a node~$\nu$ with poles~$x, y$ has \emph{\ishape} when the height of~$G(\nu)$ is less than~$\ell(y)$.
The following theorem sets the stage to prove that after this decomposition, the children of P-nodes can be arbitrarily permuted.

\begin{theorem}
    Let~$G$ be a biconnected single-source graph with unique apex~$t$.
    There exists a decomposition tree~$\mathcal T$ that represents all level-planar embeddings of~$G$ such that all children of P-nodes in~$\mathcal T$ have~\ishape.
    \label{thm:p-node-i-shape}
\end{theorem}

We see that this property ensures that P-nodes in our decomposition of level-planar graphs work analogously to those of SPQR-trees for planar graphs.
Namely, if we have a level-planar embedding~$\Gamma$ of~$G$ and consider a new embedding~$\Gamma'$ that is obtained from~$\Gamma$ by reordering the children of P-nodes, then also~$\Gamma'$ is level-planar.
We show that the~$st$-augmentation from Lemma~\ref{lem:st-augmentation} can be assumed to have certain useful properties.
The proof that the children of P-nodes can be arbitrarily permuted then uses Lemma~\ref{lem:st-lp-is-p} and the fact that the children of P-nodes in SPQR-trees can be arbitrarily permuted.
\begin{lemma}
    Let~$\Gamma$ be a level-planar embedding of~$G = (V, E)$ and let~$\mu$ be a node of~$\mathcal T$ with poles~$u, v$ so that~$G(\mu)$ has \ishape.
    Then there exists a planar~$st$-augmentation~$G' = (V, E \cup E_{st})$,~$\Gamma'$ of~$G$ and~$\Gamma$ so that~$u, v$ separates~$V(G(\mu))$ from~$V \setminus V(G(\mu))$ in~$G'$.
    \label{lem:i-shape-augmentation}
\end{lemma}
\begin{proof}
    Let~$\Gamma'$ and~$G'$ be an~$st$-augmentation of~$\Gamma$ and~$G$ where~$u, v$ is not a separation pair.
    Modify~$G', \Gamma'$ so that they remain an~$st$-augmentation of~$G, \Gamma$ and no edge in~$E_{st}$ has exactly one endpoint in~$V(G(\mu))$.
    Let~$f$ be a~$\mu$-incident arc in~$\Gamma$.
    Let~$E(f)$ denote the set of augmentation edges embedded into~$f$ to obtain~$\Gamma'$.
    Call an edge~$(w, x) \in E(f)$ \emph{critical} if~$w$ or~$x$ lies in~$V(\mu)$.
    Remove all critical edges from~$\Gamma'$ and~$G$.
    Note that because~$u, v$ is a separation pair in~$G$, the endpoints of all critical edges are now incident to the same face~$f'$.
    Observe that~$v$ is also incident to~$f'$.
    Consider a critical edge~$(w, x)$ that was removed.
    Because~$G(\mu)$ has \ishape, it follows from~$w \in V(\mu)$ that it is certainly~$\ell(w) < \ell(v)$.
    If it is~$w \not\in V(\mu)$, then it must be~$x \in V(\mu)$ and certainly~$\ell(x) < \ell(v)$.
    With~$\ell(w) < \ell(x)$ it follows that~$\ell(w) < \ell(v)$.
    So for each critical edge~$(w, x)$ the non-critical edge~$(w, v)$ can be added to~$\Gamma'$ and~$G'$.
    Because all endpoints are incident to~$f'$ and all inserted edges share the endpoint~$v$ this preserves the planarity of~$\Gamma'$ and~$G'$.
    Therefore,~$\Gamma'$ and~$G'$ is now an~$st$-augmentation of~$\Gamma$ and~$G$ once more.
    Finally,~$u$ and~$v$ separate~$V(G(\mu))$ from~$V \setminus V(G(\mu))$ in~$G'$ because~$G'$ contains no critical edge.
\end{proof}
\begin{figure}[t]
    \centering
    \includegraphics[width=\linewidth]{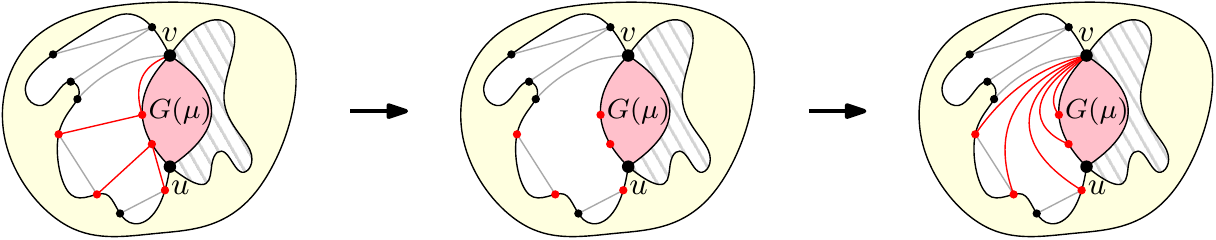}
    \caption{
        The three steps in the proof of Lemma~\ref{lem:i-shape-augmentation}.
        The subgraph~$G(\mu)$ is drawn in pink, the~$\mu$-incident face~$f$ is drawn in white.
        Critical augmentation edges in~$E(f)$ are drawn in red, and non-critical augmentation edges are drawn in gray.
        In the first step, remove all critical edges, this gives the drawing in the middle.
        Note that the red vertices and~$v$ are incident to a shared face.
        Finally, attach all red vertices to~$v$, this gives the drawing on the right.
        The same process would then be repeated for the other~$\mu$-incident face, drawn with gray stripes.
    }
    \label{fig:i-shape-augmentation}
\end{figure}

This sets the stage for the correctness proof.
The idea is to transform any given~$st$-augmentation to one that satisfies the conditions from Lemma~\ref{lem:i-shape-augmentation}.
Then the graphs corresponding to child virtual edges can be permuted arbitrarily while preserving planarity.
Lemma~\ref{lem:st-lp-is-p} then gives that all these embeddings are also level planar.
\begin{lemma}
    Let~$\Gamma$ be a level-planar embedding of~$G$ and let~$\mathcal T$ be a decomposition tree of~$G$ whose skeletons are embedded according to~$\Gamma$.
    Further, let~$\mu$ be a P-node of~$\mathcal T$.
    Let~$\Gamma'$ be the planar embedding obtained by arbitrarily permuting the child virtual edges of~$\mu$.
    Then~$\Gamma'$ is level planar.
\end{lemma}
\begin{proof}
    Let~$\Gamma'$ and~$G'$ be an~$st$-augmentation obtained from~$\Gamma$ and~$G$ according to Lemma~\ref{lem:i-shape-augmentation}.
    Note that~$(u, v)$ separates~$G'(\nu)$ from the rest of~$G'$ for each child~$\nu$ of~$\mu$.
    Consider the SPQR-tree~$\mathcal T'$ of~$G'$.
    Then~$u, v$ are the poles of a P-node~$\mu'$ in~$\mathcal T'$ with the same neighbors as~$\mu$ in~$\mathcal T$.
    Then the child virtual edges of~$\skel(\mu')$ can be arbitrarily permuted to obtain a planar embedding.
    Because~$G'$ is an~$st$-graph, Lemma~\ref{lem:st-lp-is-p} gives that any planar embedding of~$G'$ is also level planar.
\end{proof}

This completes the proof that in our decomposition the children of P-nodes can be arbitrarily permuted.
\begin{theorem}
    \label{thm:towards-correctness}
    Let~$G$ be a biconnected single-source graph with a unique apex.
    There exists a decomposition tree~$\mathcal T$ that
    \begin{enumerate*}[label=(\roman*)]
        \item represents all level-planar embeddings of~$G$ (plus some planar, non-level-planar ones), and
        \item if all skeletons of the nodes of~$\mathcal T$ are embedded so that contracting all arcs of~$\mathcal T$ yields a level-planar embedding, then the children of all P-nodes in~$\mathcal T$ can be arbitrarily permuted and then contracting all arcs of~$\mathcal T$ still yields a level-planar embedding of~$G$.
    \end{enumerate*}
\end{theorem}

\subsection{Arc Processing}
\label{ssec:arc-processing}

In this section, we finish the construction of the LP-tree.
The basis of our construction is the decomposition tree~$\mathcal T$ from Theorem~\ref{thm:p-node-i-shape}, which represents a subset of the planar embeddings of~$G$ that contains all level-planar embeddings, and moreover all children of P-nodes have~\ishape.
We now restrict~$\mathcal T$ even further until it represents exactly the level-planar embeddings of~$G$.
As of now, all R-node skeletons have a planar embedding that is unique up to reflection, as they are either triconnected or consist of only three parallel edges.
By assumption,~$G$ is level-planar, and there exists a level-planar embedding~$\Gamma$ of~$G$.
Recall that our definition of level-planar embeddings involves demands.
Computing a level-planar embedding~$\Gamma$ of~$G$ with demands reduces to computing a level-planar embedding of the supergraph~$G'$ of~$G$ obtained from~$G$ by attaching to each vertex~$v$ of~$G$ with~$d(v) > \ell(v)$ an edge to a vertex~$v'$ with~$\ell(v') = d(v)$ without demands.
Because~$G'$ is a single-source graph whose size is linear in the size of~$G$ this can be done in linear time~\cite{dn-hpt-88}.
We equip the skeleton of each node~$\mu$ with the reference embedding~$\Gamma_\mu$ such that contracting all arcs yields the embedding~$\Gamma$.
For the remainder of this section we will work with the arc-based embedding representation.
%
As a first step, we contract any arc~$(\lambda,\mu)$ of~$\mathcal T$ where~$\lambda$ is an R-node and~$\mu$ is an S-node and label the resulting node as an R-node.
Note that, since S-nodes do not offer any embedding choices, this does not change the embeddings that are represented by~$\mathcal T$.
This step makes the correctness proof easier.
Any remaining arc~$(\lambda, \mu)$ of~$\mathcal T$ is contracted based upon two properties of~$\mu$, namely the height of~$G(\mu)$ and the space around~$\mu$ in the level-planar embedding~$\Gamma$, which we define next.
The resulting node is again labeled as an R-node.
Let~$\mu$ be a node of~$\mathcal T$ with poles~$u$ and~$v$.
We denote by~$\Gamma \circ \mu$ the embedding obtained from~$\Gamma$ by contracting~$G(\mu)$ to the single edge~$e=(u,v)$.
We call the faces~$f_1,f_2$ of~$\Gamma$ that induce the incident faces of~$e$ in~$\Gamma \circ \mu$ the \emph{$\mu$-incident faces}.
The \emph{space around~$\mu$ in~$\Gamma$} is~$\min\{\ell(\apex(f_1)),\ell(\apex(f_2))\}$; see Fig.~\ref{fig:height-space}.
\begin{figure}
    \centering
    \includegraphics{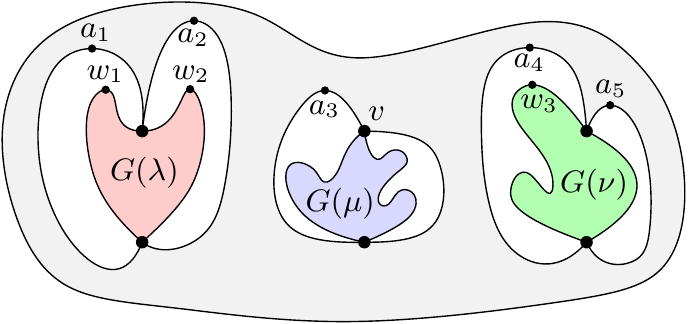}
    \caption{
        The height of~$G(\lambda)$ is at least~$\ell(w_1) = \ell(w_2)$, the height of~$G(\mu)$ is at most~$\ell(v) - 1$ and the height of~$G(\nu)$ is at least~$\ell(w_3)$.
        The space around~$\lambda$ is~$\ell(a_1)$, the space around~$\mu$ is~$\ell(v)$ and the space around~$\nu$ is~$\ell(a_5)$.
    }
    \label{fig:height-space}
\end{figure}
For the time being we will consider the embeddings of P-node skeletons as fixed.
Then all the remaining embedding choices are done by choosing whether or not to flip the embedding for the incoming arc of each R-node.
Let~$A$ denote the set of arcs in~$\mathcal T$.
For each arc~$a = (\lambda,\mu) \in A$ let~$\spc(\mu)$ denote the space around~$\mu$ in~$\Gamma$.
We label~$a$ as \emph{rigid} if~$d(\mu) \ge \spc(\mu)$ and as \emph{flexible} otherwise.

Let~$\mathcal T'$ be the decomposition tree obtained by contracting all rigid arcs and equipping each R-node skeleton with the reference embedding obtained from the contractions.
We now release the fixed embedding of the P-nodes, allowing to permute their children arbitrarily.
The resulting decomposition tree is called the \emph{LP-tree} of the input graph~$G$.
See Fig.~\ref{fig:lp-tree}~(d) for an example.
\begin{figure}
    \centering
    \includegraphics[width=\linewidth]{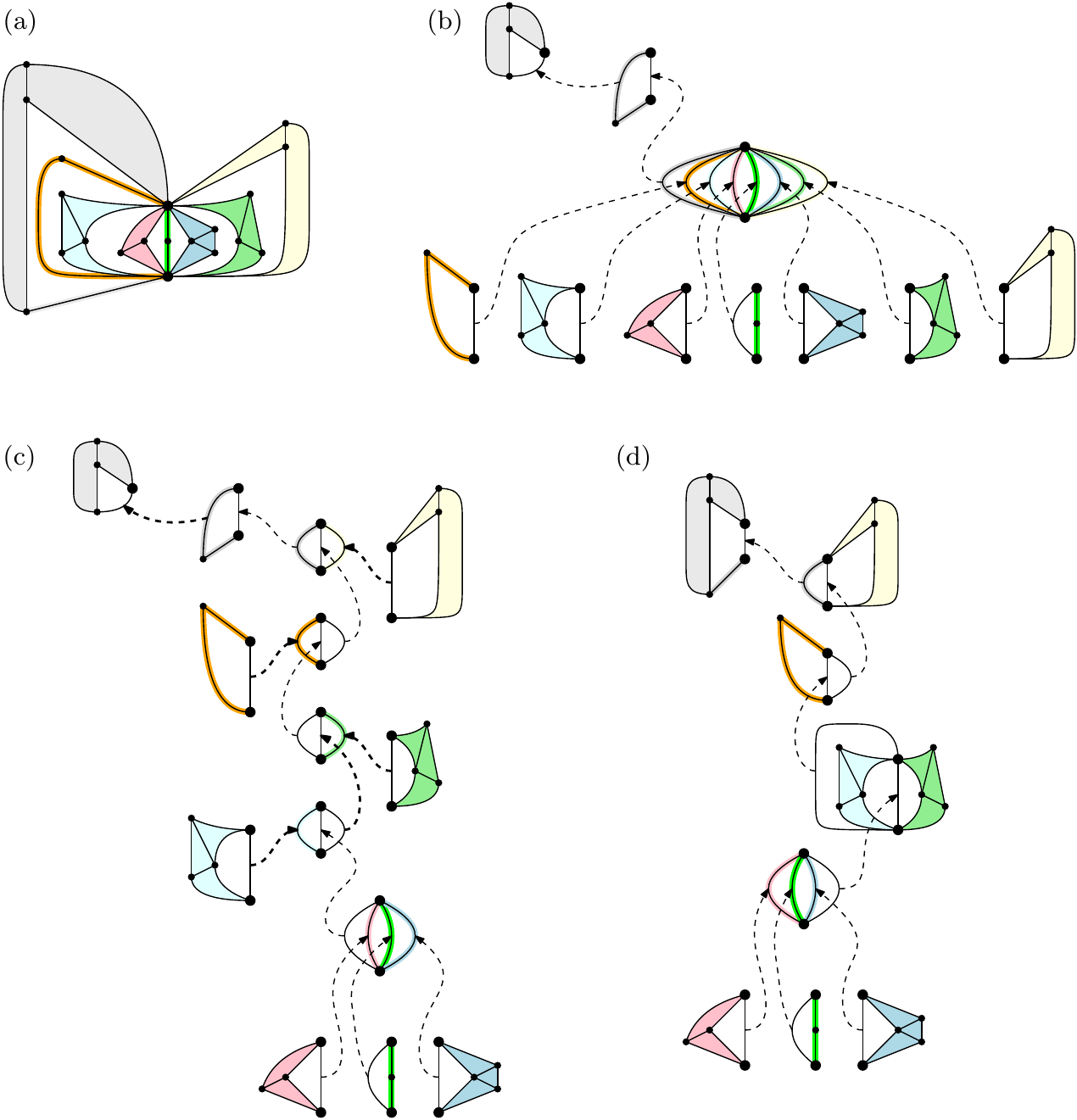}
    \caption{
        Example construction of the LP-tree for the graph~$G$~(a).
        We start with the SPQR-tree of~$G$~(b).
        Arcs are oriented towards the root.
        Next, we split the P-node, obtaining the tree shown in (c).
        Finally, we contract arcs that connect R-nodes with S-nodes and arcs that are found to be rigid (thick dashed lines).
        This gives the final LP-tree~$\mathcal T$ for~$G$~(d).
    }
    \label{fig:lp-tree}
\end{figure}
Our main result is the following theorem.

\begin{theorem}
    Let~$G$ be a biconnected, single-source, level-planar graph.
    The LP-tree of~$G$ represents exactly the level-planar embeddings of~$G$ and can be computed in linear time.
    \label{thm:correctness}
\end{theorem}

The next subsection is dedicated to proving the correctness of Theorem~\ref{thm:correctness}.
The above algorithm considers every arc of~$\mathcal T$ once.
The height of~$\mu$ and the space around~$\mu$ in~$\Gamma$ can be computed in polynomial time.
Thus, the algorithm has overall polynomial running time.
In Section~\ref{ssec:implementation-in-linear-time}, we present a linear-time implementation of this algorithm.

\subsection{Correctness}
\label{sssec:correctness-arc-processing}

Process the arcs in top-down order~$\alpha_1,\dots,\alpha_m$.
For~$i = 0, 1, \dots, m$ let the set~$A_i = \{\alpha_1,\dots,\alpha_i\}$ contain the first~$i$ processed arcs for~$i=0,\dots,m$.
Note that~$A_0 = \emptyset$ and~$A_m = A$.
Denote by~$R_i$ and~$F_i$ the arcs in~$A_i$ that are labeled rigid and flexible, respectively.
We now introduce a refinement of the embeddings represented by a decomposition tree.
Namely, a \emph{restricted decomposition tree}~$\mathcal T$ is a decomposition tree together with a subset of its arcs that are labeled as flexible, and, in the arc-based view, the embeddings represented by~$\mathcal T$ are only those that can be created by flipping only at flexible arcs.
We denote by~$\mathcal T_i$ the restricted decomposition tree obtained from~$\mathcal T$ by marking only the edges in~$F_i$ as flexible.

Initially,~$F_0 = \emptyset$, and therefore~$\mathcal T$ represents exactly the reference embedding~$\Gamma_\mathrm{ref}$ and its reflection.
Since all children of~$P$-nodes have \ishape and each P-node has \ishape, no arc incident to a P-node is labeled \emph{rigid}.
Therefore, if such an edge is contained in~$A_i$, it is flexible.
In particular, only arcs between adjacent R-nodes are labeled rigid.
As we proceed and label more edges as \emph{flexible}, more and more embeddings are represented.
Each time, we justify the level planarity of these embeddings.
As a first step, we extend the definition of space from the previous subsection, which strongly depends on the initial level-planar embedding~$\Gamma$, in terms of all level-planar embeddings represented by the restricted decomposition tree~$\mathcal T_i$.
Let~$\mu$ be a node of~$\mathcal T_i$ with poles~$u, v$.
The \emph{space around~$\mu$} is the minimum space around~$\mu$ in any level-planar embedding represented by the restricted decomposition tree~$\mathcal T_i$.
Now let~$\Gamma$ be a planar embedding of~$G$ and let~$\Pi$ be a planar embedding of~$G(\mu)$ where~$u$ and~$v$ lie on the outer face.
Because~$u$ and~$v$ is a separation pair that disconnects~$G(\mu)$ from the rest of~$G$ and~$G(\mu)$ is connected, the embedding of~$G(\mu)$ in~$\Gamma$ can be replaced by~$\Pi$.
Let~$\Gamma + \Pi$ refer to the resulting embedding.
Now let~$\Gamma$ be a planar embedding of~$G$ and let~$\mu$ be a node of~$\mathcal T$.
Let~$\Pi$ denote the restriction of~$\Gamma$ to~$G(\mu)$ and let~$\bar\Pi$ be the reflection of~$\Pi$.
\emph{Reflecting}~$\mu$ in~$T$ corresponds to replacing~$\Pi$ by~$\bar \Pi$ in~$\Gamma$, obtaining the embedding~$\Gamma + \bar\Pi$ of~$G$.

The idea is to show that if there is (is not) enough space around a node~$\mu$ to reflect it, it can (cannot) be reflected regardless of which level-planar embedding is chosen for~$G(\mu)$.
So, the algorithm always labels arcs correctly.
We use the following invariant.

\begin{lemma}
    The restricted decomposition tree~$\mathcal T_i$ satisfies the following five conditions.
    \begin{enumerate}[topsep=0pt,itemsep=0pt]
        \item \label{itm:decomposition-tree-invariant-condition-one}
              All embeddings represented by~$\mathcal T_i$ are level planar.
        \item \label{itm:decomposition-tree-invariant-condition-two}
              Let~$(\lambda, \mu)$ be an arc that is labeled as flexible.
              Let~$\Gamma$ be an embedding represented by~$\mathcal T_{i - 1}$ and let~$\Pi$ be any level-planar embedding of~$G(\mu)$.
              Then~$\Gamma + \Pi$ and~$\Gamma + \bar\Pi$ are level planar.
        \item \label{itm:decomposition-tree-invariant-condition-three}
              Let~$(\lambda, \mu)$ be an arc that is labeled as rigid.
              Let~$\Gamma$ be an embedding represented by~$\mathcal T_{i - 1}$ and let~$\Pi$ be a level-planar embedding of~$G(\mu)$ so that~$\Gamma + \Pi$ is level planar.
              Let all skeletons of~$\mathcal T_i$ be embedded according to~$\Gamma + \Pi$.
              Then~$\skel(\mu)$ has the reference embedding and~$\Gamma + \bar\Pi$ is not level planar.
        \item \label{itm:decomposition-tree-invariant-condition-four}
              The space around each node~$\mu$ of~$\mathcal T_i$ is the same across all embeddings represented by~$\mathcal T_i$.
        \item \label{itm:decomposition-tree-invariant-condition-five}
              Let~$\Gamma$ be a level-planar embedding of~$G$ so that there exists a level-planar embedding~$\Gamma_p$ of~$G$ that
              \begin{enumerate*}[label=(\roman*)]
                  \item is obtained from~$\Gamma$ by reordering the children of P-nodes, and
                  \item satisfies~$\Gamma_p = \Gamma_\mathrm{ref}(\pi_1, \pi_2, \ldots, \pi_m)$ where~$\pi_j$ indicates whether arc~$\alpha_j = (\lambda_j, \mu_j)$ should be flipped ($\pi_j = \bar\alpha_j$) or not~($\pi_j = \alpha_j$), and it is~$\pi_j = \alpha_j$ for~$j > i$.
              \end{enumerate*}
              Then~$\Gamma$ is represented by~$\mathcal T_i$.
    \end{enumerate}
    \label{lem:decomposition-tree-invariant}
\end{lemma}

\begin{proof}
    For~$i = 0$, no arc of the restricted decomposition tree~$\mathcal T_0$ is labeled as flexible.
    So~$\mathcal T_0$ only represents the reference embedding~$\Gamma_\mathrm{ref}$ and its reflection~$\bar\Gamma_\mathrm{ref}$.
    Both of these are level planar by assumption, so condition~\ref{itm:decomposition-tree-invariant-condition-one} is satisfied.
    Because~$A_0 = \emptyset$, no arc has been labeled as flexible or rigid, so conditions~\ref{itm:decomposition-tree-invariant-condition-two} and~\ref{itm:decomposition-tree-invariant-condition-three} are trivially satisfied.
    Because the incidences of vertices and faces are the same in~$\Gamma$ and its reflection~$\bar\Gamma$, condition~\ref{itm:decomposition-tree-invariant-condition-four} is also satisfied.

    Now consider the case~$i \geq 1$.
    Let~$\alpha_i = (\lambda, \mu)$.
    Let~$u, v$ be the poles of~$\mu$.
    Let~$\Gamma$ be an embedding represented by~$\mathcal T_{i - 1}$ and let~$\Pi$ be any level-planar embedding of~$G(\mu)$.
    Consider the embedding~$\Gamma + \Pi$.
    Let~$f_1, f_2$ be the~$\mu$-incident faces of~$\Gamma + \Pi$.
    For~$j = 1, 2$, let~$W_j$ be the subset of vertices of~$G(\mu)$ that are incident to~$f_j$, except for~$u$ and~$v$.
    And let~$V_j$ be all other vertices incident to~$f_j$, including~$u$ and~$v$.
    Now consider the embedding~$\Gamma + \bar\Pi$.
    Again, let~$f_1', f_2'$ be the~$\mu$-incident faces of~$\Gamma + \bar\Pi$.
    Then~$V_1 \cup W_2$ and~$V_2 \cup W_1$ are the set of vertices incident to~$f_1'$ and~$f_2'$, respectively.
    Note that all faces in~$\Gamma + \bar\Pi$ except for~$f_1', f_2'$ appear identically in~$\Gamma + \Pi$.
    Let~$a_1$ and~$a_2$ denote the apices of~$f_1$ and~$f_2$, respectively.
    Then the space around~$\mu$ in~$\Gamma + \Pi$, denoted by~$\spc(\mu)$, is~$\min(\ell(a_1), \ell(a_2))$.
    Distinguish two cases, namely~$\height(\mu) < \spc(\mu)$ and~$\height(\mu) \ge \spc(\mu)$.
    Note that because of condition~\ref{itm:decomposition-tree-invariant-condition-four}, the same case applies for any embedding represented by~$\mathcal T_{i - 1}$.

    \begin{enumerate}
        \item Consider the case~$\height(\mu) < \spc(\mu)$.
              This implies~$a_1 \in V_1$ and~$a_2 \in V_2$.
              We have to show that both~$\Gamma + \Pi$ and~$\Gamma + \bar\Pi$ are level planar.
              To this end, use Lemma~\ref{lem:level-planarity-characterisation}.
              By assumption,~$\Pi$ is a level-planar embedding of~$G(\mu)$.
              So the condition of Lemma~\ref{lem:level-planarity-characterisation} is satisfied for any vertex of~$G$ whose incident faces are all inner faces of~$\Pi$ in~$\Gamma + \Pi$ (or of~$\bar\Pi$ in~$\Gamma + \bar\Pi$).
              By condition~\ref{itm:decomposition-tree-invariant-condition-one},~$\Gamma$ is a level-planar embedding of~$G$.
              So the condition of Lemma~\ref{lem:level-planarity-characterisation} is satisfied for any vertex of~$G \setminus G(\mu)$ that is not incident to~$f_1$ and~$f_2$.
              It remains to be shown that the condition of Lemma~\ref{lem:level-planarity-characterisation} is satisfied for the vertices in~$(V_1 \setminus \{a_1\} \cup (V_2 \setminus \{a_2\}) \cup W_1 \cup W_2$.

              \begin{itemize}
                  \item Suppose~$w \in V_1 \setminus \{a_1\}$.
                        Then~$w$ is incident to~$f_1'$, as are the vertices in~$V_1$.
                        In particular, because~$a_1 \in V_1$, the apex~$a_1$ is incident to~$f_1'$.
                        And because~$a_1$ is the unique apex of~$f_1$, it is~$\ell(w) < \ell(a_1)$.
                        The argument works analogously for~$w \in V_2 \setminus \{a_2\}$.
                  \item Otherwise, it is~$w \in W_1$.
                        \begin{itemize}
                            \item Consider~$\Gamma + \Pi$.
                                  Then~$w$ is incident to~$f_1'$, as are the vertices in~$V_1$.
                                  In particular,~$a_1$ is incident to~$f_1'$.
                                  Note that it is~$\spc(\mu) = \min(\ell(a_1), \ell(a_2))$.
                                  So it is
                                  \[
                                      \ell(w) \leq \height(\mu) < \spc(\mu) \leq \ell(a_1)
                                  \]
                                  and it follows that~$\ell(w) < \ell(a_1)$.
                            \item Consider~$\Gamma + \bar\Pi$.
                                  Then~$w$ is incident to~$f_2'$, as are the vertices in~$V_2$.
                                  In particular,~$a_2$ is incident to~$f_2'$.
                                  Note that it is~$\spc(\mu) = \min(\ell(a_1), \ell(a_2))$.
                                  So it is
                                  \[
                                      \ell(w) \leq \height(\mu) < \spc(\mu) \leq \ell(a_2)
                                  \]
                                  and it follows that~$\ell(w) < \ell(a_2)$.
                        \end{itemize}
                        The argument works analogously for~$w \in W_2$.
              \end{itemize}

              This shows that the condition in Lemma~\ref{lem:level-planarity-characterisation} is satisfied for all vertices in~$\Gamma + \Pi$ and~$\Gamma + \bar\Pi$.
              As a result, both of these embeddings are level planar.

        \item Consider the case~$\height(\mu) \ge \spc(\mu)$.
              Then the algorithm will find the arc~$\alpha_i$ to be rigid and we have to show that this is the correct choice.
              Note that as observed above, the fact that~$\alpha_i$ is labeled as rigid means that~$\mu$ is an R-node.
              Recall that~$u, v$ are the poles of~$\mu$ and let~$w \neq u, v$ be a vertex of~$G(\mu)$ so that~$\ell(w)$ equals~$\height(\mu)$.
              Note that it is~$w \neq v$ by definition of~$\height(\mu)$ and~$w \neq u$ because of~$\height(\mu) \ge \spc(\mu)$.
              Again, because of~$\height(\mu) \ge \spc(\mu)$, the apex~$w$ lies on the outer face of~$\Pi$.
              Either~$w$ is a vertex on the outer face of~$\skel(\mu)$, or~$w$ belongs to~$G(e)$ for some child virtual edge~$e$ on the outer face of~$\skel(\mu)$.
              Because~$\mu$ is an R-node, its skeleton is biconnected and therefore~$w$ is incident to either~$f_1$ or~$f_2$, but not both, and this choice depends entirely on the embedding of~$\skel(\mu)$.
              By assumption~$\Gamma + \Pi$ is level planar and it remains to be shown that~$\Gamma + \bar\Pi$, is not level planar.
              Note that~$\Gamma + \bar\Pi$ is the embedding that is obtained by reflecting~$\mu$ so that~$\skel(\mu)$ does not have the reference embedding.
              Assume~$w \in W_1$ without loss of generality.
              It is~$\ell(w) = \max\{\ell(x) \mid x \in W_1\}$.
              Because~$w$ is an apex of~$V(\mu)$, face~$f_1$ must be the face incident to~$w$ of which~$w$ is not an apex.
              Now consider~$\Gamma + \bar\Pi$.
              Now~$w$ is incident to face~$f_2'$ which is incident to the vertices~$V_2 \cup W_1$.
              Because~$\height(\mu) \ge \spc(\mu)$ it is~$\ell(w) \geq \max\{w \in V_2 \cup W_1\}$.
              This means that~$w$ is an apex of all its incident faces.
              Then~$\Gamma + \bar\Pi$ cannot be level planar by Lemma~\ref{lem:level-planarity-characterisation}.
    \end{enumerate}

    This means that if~$a$ is labeled as flexible, then~$G(\mu)$ can be reflected in all embeddings represented by~$\mathcal T_{i - 1}$.
    And if~$a$ is labeled as rigid, then~$G(\mu)$ cannot be reflected in any embedding represented by~$\mathcal T_{i - 1}$.
    This shows that~$\mathcal T_i$ satisfies conditions~\ref{itm:decomposition-tree-invariant-condition-one} through~\ref{itm:decomposition-tree-invariant-condition-three}.
    Next, we show that the space around nodes of~$\mathcal T$ is the same across all embeddings represented by~$\mathcal T_i$.
    Once again, distinguish the two cases~$\height(\mu) < \spc(\mu)$ and~$\height(\mu) \ge \spc(\mu)$.

    \begin{enumerate}
        \item Consider the case~$\height(\mu) < \spc(\mu)$.
              Let~$\Gamma$ be an embedding represented by~$\mathcal T_{i - 1}$ and let~$\Gamma'$ be the embedding obtained by reflecting~$\mu$ in~$\Gamma$.
              See Fig.~\ref{fig:space-no-change}.
              We show that the space around each node~$\nu$ of~$\mathcal T_i$ is identical in~$\Gamma$ and~$\Gamma'$.
              Let~$x, y$ be the poles of~$\nu$ and let~$f_1, f_2$ be the~$\mu$-incident faces in~$\Gamma$.
              Further, let~$f_1', f_2'$ be the~$\mu$-incident faces in~$\Gamma'$.
              As previously discussed, all faces in~$\Gamma$ and~$\Gamma'$ are identical, except for~$f_1, f_2, f_1', f_2'$.
              Suppose that both~$\nu$-incident faces in~$\Gamma$ are neither~$f_1$ nor~$f_2$.
              Then the faces around~$\nu$ do not change and therefore the space around~$\nu$ does not change.
              Conversely, suppose that the~$\nu$-incident faces are~$f_1$ and~$f_2$.
              Then the space around~$\nu$ in~$\Gamma$ is~$\min(\ell(a_1), \ell(a_2))$.
              And because~$a_j \in V_j$ for~$j = 1, 2$, the space around~$\nu$ in~$\Gamma'$ is~$\min(\ell(a_1), \ell(a_2))$ as well.

              Otherwise, exactly one~$\nu$-incident face in~$\Gamma$ is either~$f_1$ or~$f_2$.
              Without loss of generality, let~$f_1$ be that face.
              Then exactly one~$\nu$-incident face in~$\Gamma'$ is either~$f_1'$ or~$f_2'$.
              Assume that face is~$f_1'$.
              Because the apex of~$f_1'$ and~$f_1$ are identical, the space around~$\nu$ in~$\Gamma'$ is the same as in~$\Gamma$.
              Now assume that~$f_2'$ is the face.
              Then the space around~$\nu$ is bounded by a vertex~$z \in V(G(\mu))$ and~$\height(\mu) < \spc(\mu)$ implies that~$\ell(z) \leq \height(\mu) < \spc(\mu)$.
              So the space around~$\nu$ is bounded by~$z$ in~$\Gamma$ and~$\Gamma'$.

              Again, because of condition~\ref{itm:decomposition-tree-invariant-condition-four}, the argument can be made for any embedding represented by~$\mathcal T_{i - 1}$, and therefore the claim follows for all embeddings represented by~$\mathcal T_i$.
        \item Consider the case~$\height(\mu) \ge \spc(\mu)$.
              Then~$\mathcal T_i$ represents the same embeddings as~$\mathcal T_{i - 1}$ and so condition~\ref{itm:decomposition-tree-invariant-condition-four} is trivially satisfied.
    \end{enumerate}

    \begin{figure}[t]
        \centering
        \includegraphics{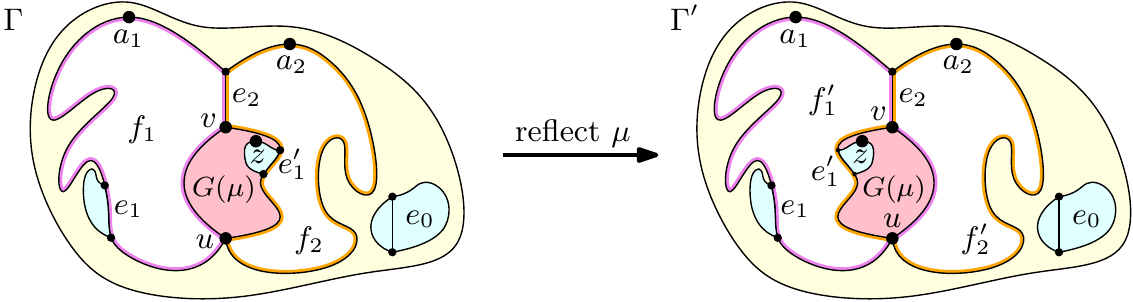}
        \caption{
            Proof of Lemma~\ref{lem:decomposition-tree-invariant}, Property~\ref{itm:decomposition-tree-invariant-condition-four} for the case~$\height(\mu) < \spc(\mu)$.
            The edge~$e_0$ is not incident to any~$\mu$-incident face, the edges~$e_1, e_1'$ are incident to exactly one~$\mu$-incident face and the edge~$e_2$ is incident to both~$\mu$-incident faces.
            The space around all nodes of~$\mathcal T$ does not change when reflecting~$\mu$.
        }
        \label{fig:space-no-change}
    \end{figure}

    Now we show that permuting the children of P-nodes does not change the space around any node of~$\mathcal T_i$.
    Recall Theorem~\ref{thm:p-node-i-shape}, which states that all children of P-nodes have \ishape.
    Take any two adjacent children of a P node~$\mu$ and merge them, creating a new R-node child~$\nu$ of the P-node.
    Then~$\nu$ has \ishape.
    Therefore it can be reflected.
    Further reflecting both children of~$\nu$, which is possible because they too have \ishape, means that in the resulting embedding the two children are reversed.
    Note that any permutation can be realized by a number of exchanges of adjacent pairs, which shows that condition~\ref{itm:decomposition-tree-invariant-condition-four} remains satisfied when permuting the children of P-nodes.
    This shows that condition~\ref{itm:decomposition-tree-invariant-condition-four} is satisfied for~$\mathcal T_i$.

    As the final step, we prove that condition~\ref{itm:decomposition-tree-invariant-condition-five} is satisfied for~$\mathcal T_i$.
    Recalling Theorem~\ref{thm:towards-correctness} and the equivalence of the skeleton-based and arc-based representations, we have that for every level-planar embedding~$\Gamma$ of~$G$ there exists a level-planar embedding~$\Gamma_p$ that is obtained from~$\Gamma$ by reordering the children of P-nodes such that it is~$\Gamma_p = \Gamma_\mathrm{ref}(\pi_1, \pi_2, \ldots, \pi_m)$ where it is~$\alpha_j = (\lambda_j, \mu_j)$ and~$\pi_j = \alpha_j$ or~$\pi_j = \bar\alpha_j$ denotes whether the embedding of~$G(\mu_i)$ should remain unchanged or be flipped, respectively.
    Now let~$\pi_j = \alpha_j$ for~$j > i$ as required by the invariant.
    We show that~$\Gamma_p$ is represented by~$\mathcal T_i$, Theorem~\ref{thm:towards-correctness} then implies that~$\Gamma$ is represented by~$\mathcal T_i$ as well.

    In the base case~$i = 0$ no arc is flipped, i.e., we have~$\Gamma_p = \Gamma_\mathrm{ref}$, which is the level-planar embedding of~$G$ represented by~$\mathcal T_0$ by definition.
    In the inductive case~$i > 0$, we distinguish two cases based on whether it is~$\pi_i = \alpha_i$ or~$\pi_i = \bar\alpha_i$.
    Define
    \begin{align*}
        \Gamma_p^1 &= \Gamma_\mathrm{ref}(\pi_1 ,\pi_2, \ldots, \pi_{i - 1}, \alpha_i, \alpha_{i + 1}, \ldots, \alpha_m) \quad\text{ and}\\
        \Gamma_p^2 &= \Gamma_\mathrm{ref}(\pi_1 ,\pi_2, \ldots, \pi_{i - 1}, \neg\alpha_i, \alpha_{i + 1}, \ldots, \alpha_m).
    \end{align*}
    Observe that~$\Gamma_p^1$ is represented by~$\mathcal T_{i - 1}$ by induction on condition~\ref{itm:decomposition-tree-invariant-condition-five}.
    Then~$\Gamma_p^1$ is also represented by~$\mathcal T_i$, which shows the claim for~$\Gamma_p = \Gamma_p^1$.
    Otherwise, it is~$\Gamma_p = \Gamma_p^2$. 
    Let~$\Pi$ denote the restriction of~$\Gamma_p^1$ to~$G(\mu_i)$.
    Then it is~$\Gamma_p^1 = \Gamma_p^1 + \Pi$ and flipping~$\alpha_i$ reflects~$\Pi$, i.e.,~$\Gamma_p^2 = \Gamma_p^1 + \bar\Pi$.
    We now distinguish two cases based on whether~$\alpha_i$ is labeled as flexible or rigid.
    If~$\alpha_i$ is labeled as flexible,~$\Gamma_p^2 = \Gamma_p$ is represented by~$\mathcal T_i$.
    Otherwise,~$\alpha_i$ is labeled as rigid.
    Recall that~$\Gamma_p^1$ is represented by~$\mathcal T_{i - 1}$ and~$\mathcal T_i$.
    Then condition~\ref{itm:decomposition-tree-invariant-condition-three} gives that~$\Gamma_p = \Gamma_p^2 = \Gamma_p^1 + \bar\Pi$ is not level planar, a contradiction.
\end{proof}

\noindent
The restricted decomposition tree~$\mathcal T_m$ represents only level-planar embeddings by Property~\ref{itm:decomposition-tree-invariant-condition-one} of Lemma~\ref{lem:decomposition-tree-invariant}.
Because no arc of~$\mathcal T_m$ is unlabeled, it also follows that all level-planar embeddings of~$G$ are represented by~$\mathcal T_m$.
Contracting all arcs labeled as rigid in~$\mathcal T_m$ gives the LP-tree for~$G$, which concludes our proof of Theorem~\ref{thm:correctness}.

\subsection{Construction in Linear Time}
\label{ssec:implementation-in-linear-time}

The algorithm described in Section~\ref{sec:decomposition} clearly has polynomial running time.
In this section, we describe an implementation of it that has linear running time.
Starting out, the preprocessing step where the apex~$t$ and the edge~$(s, t)$ is added to~$G$ is feasible in linear time.
Next, the SPQR-tree~$\mathcal T$ of this modified graph~$G$ can be computed in linear time~\cite{gm-altiost-01,ht-dagitc-73}.
Then, a level-planar embedding~$\Gamma$ of~$G$ is computed in linear time~\cite{dn-hpt-88} and all skeletons of~$\mathcal T$ are embedded accordingly.

For each node~$\mu$ of~$\mathcal T$ the height of~$G(\mu)$ needs to be known.
The heights for all nodes are computed bottom-up.
Note that the height of an edge~$e = (u, v)$ of~$G$ is~$\ell(u)$.
This means that the heights for all leaf Q-nodes can be easily determined.
In general, to determine the height for a node~$\mu$ of~$\mathcal T$, proceed as follows.
Assume the heights are known for all children.
Let~$E_\mu$ be the child virtual edges of~$\skel(\mu)$ and let~$h(e)$ denote the height of~$G(\nu)$ with~$\corr_\mu(e) = \nu$.
Then the height of~$\mu$ is~$\max \{ \{ d(e) \mid e \in E_\mu \} \cup \{ d(w) \mid w \in V(\mu) \} \}$.
Thus, the running time spent to determine the height of~$\mu$ when the heights of all its children is known is linear in the size of~$\skel(\mu)$.
Because the sum of the sizes of all skeletons of~$\mathcal T$ is linear in~$n$, all heights can be computed in linear time.

The next step is to split P-nodes.
Let~$\mu$ be a P-node.
One split at~$\mu$ requires to find the child with the greatest height.
Because~$\Gamma$ is a level-planar embedding, Lemma~\ref{lem:p-node-heights} gives that this is one of the outermost children.
By inspecting the two outermost children of~$\mu$, the child~$\nu$ with greatest height can be found, or it is found that all children of~$\mu$ have \ishape and~$\mu$ does not need to be split.
A P-node split is a constant-time operation.
Because there are no more P-node splits than nodes in~$\mathcal T$, all P-node splits are feasible in linear time.

The final step of the algorithm is to process all arcs.
For this the space around each node needs to be known.
The space around a node~$\mu$ depends on the apices of the~$\mu$-incident faces in~$\Gamma$.
Fortunately, these can be easily computed bottom-up.
Start by labeling every face~$f$ of~$\Gamma$ with its apex by walking around the cycle that bounds~$f$.
For every edge~$e$ of~$G$ the apices on both sides of~$e$ can then be looked up in~$\Gamma$.
So the incident apices are known for each Q-node of~$\mathcal T$.
Let~$\mu$ be a node of~$\mathcal T$ so that for each child~$\nu$ of~$\mu$ the apices of the~$\nu$-incident faces are known.
Then the apices of the~$\mu$-incident faces can be determined from the child virtual edges of~$\skel(\mu)$ that share a face with the parent virtual edge of~$\mu$.
The running time of this procedure is linear in the sum of sizes of all skeletons, i.e., linear in~$n$.
To process the arcs, simply walk through~$\mathcal T$ from the top down.
Compute the space around each child node~$\nu$ from the available apices of the~$\nu$-incident faces and compare it with the precomputed height of~$G(\mu)$.
Finally, contract all arcs marked as rigid, which again is feasible in overall linear time.
This proves the running time claimed in Theorem~\ref{thm:correctness}.

\section{Applications}
\label{sec:applications}

We use the LP-tree to translate efficient algorithms for constrained planarity problems to the level-planar setting.
First, we extend the partial planarity algorithm by Angelini et al.~\cite{adfjk-tppeg-15} to solve partial level planarity for biconnected single-source level graphs.
Second, we adapt this algorithm to solve constrained level planarity.
In both cases we obtain a linear-time algorithm, improving upon the best previously known running time of~$O(n^2)$, though that algorithm also works in the non-biconnected case~\cite{br-pclp-17}.
Third, we translate the simultaneous planarity algorithm due to Angelini et al.~\cite{adfpr-ttseotg-12} to the simultaneous level planarity problem when the shared graph is a biconnected single-source level graph.
Previously, no polynomial-time algorithm was known for this problem.

\subsection{Partial Level Planarity}
\label{ssec:partial-level-planarity}

Angelini et al.~define partial planarity in terms of the cyclic orders of edges around vertices (the ``edge-order definition'') as follows.
A partially embedded graph~(\textsc{Peg}) is a triple~$(G, H, \mathcal H)$ that consists of a graph~$G$ and a subgraph~$H$ of~$G$ together with a planar embedding~$\mathcal H$ of~$H$.
The task is to find an embedding~$\mathcal G$ of~$G$ that extends~$\mathcal H$ in the sense that any three edges~$e, f, g$ of~$H$ that are incident to a shared vertex~$v$ appear in the same order around~$v$ in~$\mathcal G$ as in~$\mathcal H$.
The algorithm works by representing all planar embeddings of~$G$ as an SPQR-tree~$\mathcal T$ and then determining whether there exists a planar embedding of~$G$ that extends the given partial embedding~$\mathcal H$ as follows.
Recall that~$e, f, g$ correspond to distinct Q-nodes~$\mu_e, \mu_f$ and~$\mu_g$ in~$\mathcal T$.
There is exactly one node~$\nu$ of~$\mathcal T$ that lies on all paths connecting two of these Q-nodes.
Furthermore,~$e, f, g$ belong to the expansion graphs of three distinct virtual edges~$\hat e, \hat f, \hat g$ of~$\skel(\nu)$.
The order of~$e, f$ and~$g$ in the planar embedding represented by~$\mathcal T$ is determined by the order of~$\hat e, \hat f, \hat g$ in~$\skel(\nu)$, i.e., by the embedding of~$\skel(\nu)$.
Fixing the relative order of~$e, f, g$ therefore imposes certain constraints on the embedding of~$\skel(\mu)$.
Namely, an R-node can be constrained to have exactly one of its two possible embeddings and the admissible permutations of the neighbors of a P-node can be constrained as a partial ordering.
To model the embedding~$\mathcal H$ consider for each vertex~$v$ of~$H$ each triple~$e, f, g$ of consecutive edges around~$v$ and fix their order as in~$\mathcal H$.
The algorithm collects these linearly many constraints and then checks whether they can be satisfied simultaneously.

Define partial level planarity analogously, i.e., a \emph{partially embedded level graph} is a triple~$(G, H, \mathcal H)$ of a level graph~$G$, a subgraph~$H$ of~$G$ and a level-planar embedding~$\mathcal H$ of~$H$.
Again the task is to find an embedding~$\mathcal G$ of~$G$ that extends~$\mathcal H$ in the sense that any three edges~$e, f, g$ of~$H$ that are incident to a shared vertex~$v$ appear in the same order around~$v$ in~$\mathcal G$ as in~$\mathcal H$.
This definition of partial level planarity is distinct from but~(due to Lemma~\ref{lem:lp-embedding-is-equivalence-class}~($\star$)) equivalent to the one given in~\cite{br-pclp-17}, which is a special case of constrained level planarity as presented in the next section.
LP-trees exhibit all relevant properties of SPQR-trees used by the partial planarity algorithm.
Ordered edges~$e, f, g$ of~$G$ again correspond to distinct Q-nodes of the LP-tree~$\mathcal T'$ for~$G$.
Again, there is a unique node~$\nu$ of~$\mathcal T'$ that has three virtual edges~$\hat e, \hat f, \hat g$ that determine the order of~$e, f, g$ in the level-planar drawing represented by~$\mathcal T'$.
Finally, in LP-trees just like in SPQR-trees, R-nodes have exactly two possible embeddings and the virtual edges of P-nodes can be arbitrarily permuted.
Using the LP-tree as a drop-in replacement for the SPQR-tree in the partial planarity algorithm due to Angelini et al.~gives the following, improving upon the previously known best algorithm with~$O(n^2)$ running time (although that algorithm also works for the non-biconnected case~\cite{br-pclp-17}).

\begin{theorem}
    \label{thm:partial-level-planarity}
    Partial level planarity can be solved in linear running time for biconnected single-source level graphs.
\end{theorem}

\noindent
Angelini et al.~extend their algorithm to the connected case~\cite{adfjk-tppeg-15}.
This requires significant additional effort and the use of another data structure, called the enriched block-cut tree, that manages the biconnected components of a graph in a tree.
Some of the techniques described in this paper, in particular our notion of demands, may be helpful in extending our algorithm to the connected single-source case.
Consider a connected single-source graph~$G$.
All biconnected components of~$G$ have a single source and the LP-tree can be used to represent their level-planar embeddings.
However, a vertex~$v$ of some biconnected component~$H$ of~$G$ may be a cutvertex in~$G$ and can dominate vertices that do not belong to~$H$.
Depending on the space around~$v$ and the levels on which these vertices lie this may restrict the admissible level-planar embeddings of~$H$.
Let~$X(v)$ denote the set of vertices dominated by~$v$ that do not belong to~$H$.
Set the demand of~$v$ to~$d(v) = d(X(v))$.
Computing the LP-tree with these demands ensures that there is enough space around each cutvertex~$v$ to embed all components connected at~$v$.
The remaining choices are into which faces of~$H$ incident to~$v$ such components can be embedded and possibly nesting biconnected components.
These choices are largely independent for different components and only depend on the available space in each incident face.
This information is known from the LP-tree computation.
In this way it may be possible to extend the steps for handling non-biconnected graphs due to Angelini et al.~to the level planar setting.

\subsection{Constrained Level Planarity}
\label{ssec:constrained-level-planarity}

A \emph{constrained level graph}~(\textsc{Clg}) is a tuple~$(G, \{\prec'_1, \prec'_2, \ldots, \prec'_k \})$ that consists of a~$k$-level graph~$G$ and partial orders~$\prec'_i$ of~$V_i$ for~$i = 1, 2, \ldots, k$ (the ``vertex-order definition'')~\cite{br-pclp-17}.
The task is to find a drawing of~$G$, i.e., total orders~$\prec_i$ of~$V_i$ that extend~$\prec'_i$ in the sense that for any two vertices~$u, v \in V_i$ with~$u \prec'_i v$ it is~$u \prec_i v$.

\begin{theorem}
    \label{thm:constrained-level-planarity}
    Constrained level planarity can be solved in linear running time for biconnected single-source level graphs.
\end{theorem}

\begin{proof}
Tanslate the given vertex-order constraints into edge-order constraints.
This translation is justified by Lemma~\ref{lem:lp-embedding-is-equivalence-class}.
We now show that all vertex-order constraints can be translated in linear time.
For any pair~$u, v$ with~$u \prec'_i v$ we start by finding a vertex~$w$ so that there are disjoint paths~$p_u$ and~$p_v$ from~$w$ to~$u$ and~$v$.
This can be achieved by using the algorithm of Harel and Tarjan on a depth-first-search tree~$\mathcal D$ of~$G$~\cite{ht-faffnca-84} in linear time.
Mark~$w$ with the pair~$u, v$ for the next step.
Then, we find the edges~$e$ and~$f$ of~$p_u$ and~$p_v$ incident to~$w$, respectively.
To this end, we proceed similarly to a technique described by Bl\"asius et al.~\cite{bkr-seeorpc-18}.
At the beginning, every vertex of~$G$ belongs to its own singleton set.
Proceed to process the vertices of~$G$ bottom-up in~$\mathcal D$, i.e., starting from the vertices on the greatest level.
When encountering a vertex~$w$ marked with a pair~$u, v$, find the representatives of~$u$ and~$v$, denoted by~$u'$ and~$v'$, respectively.
Observe that it is~$e = (w, u')$ and~$f = (w, v')$, and that both~$e$ and~$f$ are tree edges of~$\mathcal D$.
Then unify the sets of all of its direct descendants in~$\mathcal D$ and let~$w$ be the representative of the resulting union.
Because all union operations are known in advance we can use the linear-time union-find algorithm of Gabow and Tarjan~\cite{gt-altafascodsu-85}.
Finally, pick some incoming edge around~$w$ as~$g$, or the edge~$(s, t)$ if~$w = s$.
In this way, we translate the constraint of the form~$u \prec'_i v$ to a constraint on the order of the edges~$e, f$ and~$g$ around~$w$.
Apply this translation for each constraint in the partial orders~$\prec'_i$.

In a similar fashion we can find the node~$\nu$ of the LP-tree~$\mathcal T$ and the three virtual edges~$\hat e, \hat f$ and~$\hat g$ of~$\skel(\nu)$ so that the relative position of~$\hat e, \hat f$ and~$\hat g$ in the embedding of~$\skel(\nu)$ determines the relative position of~$e, f$ and~$g$ in the embedding represented by~$\mathcal T$.
We can the use a similar technique as the one described for partial level planarity.
\end{proof}

\subsection{Simultaneous Level Planarity}
\label{ssec:simultaneous-level-planarity}

We translate the simultaneous planarity algorithm of Angelini et al.~\cite{adfpr-ttseotg-12} to solve simultaneous level planarity for biconnected single-source graphs.
They define simultaneous planarity as follows.
Let~$G_1 = (V, E_1)$ and~$G_2 = (V, E_2)$ be two graphs with the same vertices.
The \emph{inclusive} edges~$E_1 \cap E_2$ together with~$V$ make up the intersection graph~$G_{1 \cap 2}$, or simply~$G$ for short.
All other edges are \emph{exclusive}.
The graphs~$G_1$ and~$G_2$ admit \emph{simultaneous} embeddings~$\mathcal E_1, \mathcal E_2$ if the relative order of any three distinct inclusive edges~$e, f$ and~$g$ with a shared endpoint is identical in~$\mathcal E_1$ and~$\mathcal E_2$.
The algorithm of Angelini et al.~works by building the SPQR-tree for the shared graph~$G$ and then expressing the constraints imposed on~$G$ by the exclusive edges as a \textsc{2-Sat} instance~$S$ that is satisfiable iff~$G_1$ and~$G_2$ admit a simultaneous embedding.
We give a very brief overview of the \textsc{2-Sat} constraints in the planar setting.
\begin{figure}[t]
    \centering
    \includegraphics{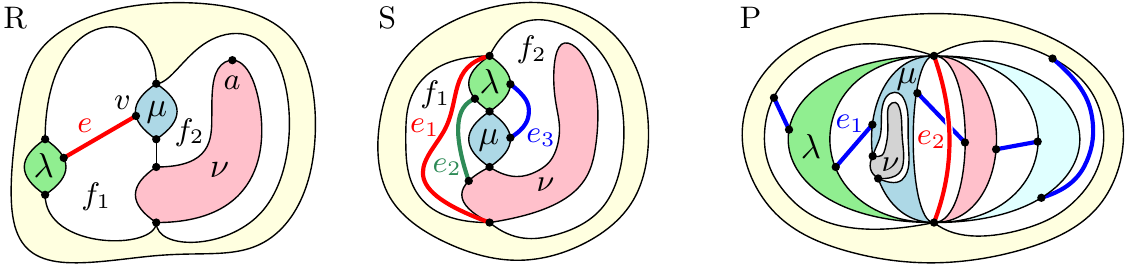}
    \caption{
        In the R-node,~$e$ fixes the relative embeddings of~$G(\lambda)$ and~$G(\mu)$.
        In the level-planar setting,~$e$ also fixes the embedding of~$G(\nu)$.
        In the S-node,~$e_2$ and~$e_3$ fix the relative embeddings of~$G(\lambda), G(\nu)$ and~$G(\lambda), G(\mu)$, respectively.
        In the level-planar setting,~$e_1$ also fixes the embedding of~$G(\nu)$.
        In the P-node,~$e_1$ fixes the relative embeddings of~$G(\lambda)$ and~$G(\mu)$.
        In the level-planar setting,~$e_1$ also fixes the embedding of~$G(\nu)$.
    }
    \label{fig:simultaneous-constraints}
\end{figure}
In an R-node, an exclusive edge~$e$ has to be embedded into a unique face.
This potentially restricts the embedding of the expansion graphs~$G(\lambda), G(\mu)$ that contain the endpoints of~$e$, i.e., the embedding of~$G(\lambda)$ and~$G(\mu)$ is fixed with respect to the embedding of the R-node.
Add a variable~$x_\mu$ to~$S$ for every node of~$\mathcal T$ with the semantics that~$x_\mu$ is true if~$\skel(\mu)$ has its reference embedding~$\Gamma_\mu$, and false if the embedding of~$\skel(\mu)$ is the reflection of~$\Gamma_\mu$.
The restriction imposed by~$e$ on~$G(\lambda)$ and~$G(\mu)$ can then be modeled as a \textsc{2-Sat} constraint on the variables~$x_\lambda$ and~$x_\mu$.
For example, in the R-node shown in Fig.~\ref{fig:simultaneous-constraints} on the left, the internal edge~$e$ must be embedded into face~$f_1$, which fixes the relative embeddings of~$G(\lambda)$ and~$G(\mu)$.
In an S-node, an exclusive edge~$e$ may be embedded into one of the two candidate faces~$f_1, f_2$ around the node.
The edge~$e$ can conflict with another exclusive edge~$e'$ of the S-node, meaning that~$e$ and~$e'$ cannot be embedded in the same face.
This is modeled by introducing for every exclusive edge~$e$ and candidate face~$f$ the variable~$x^f_e$ with the semantics that~$x^f_e$ is true iff~$e$ is embedded into~$f$.
The previously mentioned conflict can then be resolved by adding the constraints~$x^{f_1}_e \lor x^{f_2}_e$,~$x^{f_1}_{e'} \lor x^{f_2}_{e'}$ and~$x^{f_1}_e \neq x^{f_1}_{e'}$ to~$S$.
Additionally, an exclusive edge~$e$ whose endpoints lie in different expansion graphs can restrict their respective embeddings.
For example, in the S-node shown in Fig.~\ref{fig:simultaneous-constraints} in the middle, the edges~$e_2$ and~$e_3$ may not be embedded into the same face.
And~$e_2$ and~$e_3$ fix the embeddings of~$G(\lambda)$ and~$G(\nu)$ and of~$G(\lambda)$ and~$G(\mu)$, respectively.
This would be modeled as~$x_\lambda = x_\nu$ and~$x_\lambda = x_\mu$ in~$S$.
In a P-node, an exclusive edge can restrict the embeddings of expansion graphs just like in R-nodes.
Additionally, exclusive edges between the poles of a P-node can always be embedded unless all virtual edges are forced to be adjacent by internal edges.
For example, in the P-node shown in Fig.~\ref{fig:simultaneous-constraints} on the right,~$e_1$ fixes the relative embeddings of~$G(\lambda)$ and~$G(\mu)$.
And~$e_2$ can be embedded iff one of the blue edges does not exist.

Adapt the algorithm to the level-planar setting.
First, replace the SPQR-tree with the LP-tree~$\mathcal T$.
The satisfying truth assignments of~$S$ then correspond to simultaneous planar embeddings~$\mathcal E_1, \mathcal E_2$ of~$G_1, G_2$, so that their shared embedding~$\mathcal E$ of~$G$ is level planar.
However, due to the presence of exclusive edges,~$\mathcal E_1$ and~$\mathcal E_2$ are not necessarily level planar.
To make sure that~$\mathcal E_1$ and~$\mathcal E_2$ are level planar, we add more constraints to~$S$.
Consider adding an exclusive edge~$e$ into a face~$f$.
This splits~$f$ into two faces~$f', f''$.
The apex of at least one face, say~$f''$, remains unchanged.
As a consequence, the space around any virtual edge incident to~$f''$ remains unchanged as well.
But the apex of~$f'$ can change, namely, the apex of~$f'$ is an endpoint of~$e$.
Then the space around the virtual edges incident to~$f'$ can decrease.
This reduces the space around the virtual edge associated with~$\nu$.
In the same way as described in Section~\ref{ssec:arc-processing}, this restricts some arcs in~$\mathcal T$.
This can be described as an implication on the variables~$x^f_e$ and~$x_\nu$.
For an example, see Fig.~\ref{fig:simultaneous-constraints}.
In the R-node, adding the edge~$e$ with endpoint~$v$ into~$f_1$ creates a new face~$f_1'$ with apex~$v$.
This forces~$G(\nu)$ to be embedded so that its apex~$a$ is embedded into face~$f_2$.
Similarly, in the S-node and in the P-node, adding the edge~$e_1$ restricts~$G(\nu)$.
We collect all these additional implications of embedding~$e$ into~$f$ and add them to the 2-\textsc{Sat} instance~$S$.
Each exclusive edge leads to a constant number of \textsc{2-Sat} implications.
To find each such implication~$O(n)$ time is needed in the worst case.
Because there are at most~$O(n)$ exclusive edges this gives quadratic running time overall.
Clearly, all implications must be satisfied for~$\mathcal E_1$ and~$\mathcal E_2$ to be level planar.
On the other hand, suppose that one of~$\mathcal E_1$ or~$\mathcal E_2$, say~$\mathcal E_1$, is not level planar.
Because the restriction of~$\mathcal E_1$ to~$G$ is level planar due to the LP-tree and planar due to the algorithm by Angelini et al., there must be a crossing involving an exclusive edge~$e$ of~$G_1$.
This contradicts the fact that we have respected all necessary implications of embedding~$e$.
We obtain Theorem~\ref{thm:simultaneous-level-planarity}.

\begin{theorem}
    \label{thm:simultaneous-level-planarity}
    Simultaneous level planarity can be solved in quadratic time for two graphs whose intersection is a biconnected single-source level graph.
\end{theorem}

\noindent
In the non-biconnected setting Angelini et al.~solve the case when the intersection graph is a star.
Haeupler et al.~describe an algorithm for simultaneous planarity that does not use SPQR-trees, but they also require biconnectivity~\cite{hjl-tspwtcgi2c}.
The complexity of the general (connected) case remains open.

\section{Conclusion}
\label{sec:conclusion}

The majority of constrained embedding algorithms for planar graphs rely on two features of the SPQR-tree: they are decomposition trees and the embedding choices consist of arbitrarily permuting parallel edges between two poles or choosing the flip of of a skeleton whose embedding is unique up to reflection.
We have developed the LP-tree, an SPQR-tree-like embedding representation that has both of these features.
SPQR-tree-based algorithms can then usually be executed on LP-trees without any modification.
The necessity for mostly minor modifications only stems from the fact that in many cases the level-planar version of a problem imposes additional restrictions on the embedding compared to the original planar version.
Our LP-tree thus allows to leverage a large body of literature on constrained embedding problems and to transfer it to the level-planar setting.
In particular, we have used it to obtain linear-time algorithms for partial and constrained level planarity in the biconnected case, which improves upon the previous best known running time of~$O(n^2)$.
Moreover, we have presented an efficient algorithm for the simultaneous level planarity problem.
Previously, no polynomial-time algorithm was known for this problem.
Finally, we have argued that an SPQR-tree-like embedding representation for level-planar graphs with multiple sources does not substantially help in solving the partial and constrained level planarity problems, is not efficiently computable, or does not exist.

\bibliographystyle{plainurl}
\bibliography{references}

\end{document}